\documentclass[11pt,a4paper]{article}

\usepackage{fullpage}
\usepackage[utf8x]{inputenc}
\usepackage{amsmath, amsthm}
\usepackage{wrapfig,graphicx,amssymb,textcomp,array,amsmath}
\usepackage{algpseudocode} 
\usepackage{enumerate}
\algtext*{EndWhile}
\algtext*{EndIf}
\usepackage{algorithm}
\usepackage{color}

\newcommand{\CH}[1]{\text{$CH(#1)$}}
\newcommand{\HC}[1]{\text{$H(#1)$}}
\newcommand{\M}[2]{\text{$M_{#1}(#2)$}}
\newcommand{\Min}[2]{\text{\sf Min$(#1,#2)$}}
\newcommand{\Cut}[2]{\text{\sf Cut$(#1,#2)$}}
\newcommand{\Tangent}[2]{\text{\sf Tangent$(#1,#2)$}}
\newcommand{\UT}[2]{\text{\em UpperTangent$(#1,#2)$}}
\newcommand{\LT}[2]{\text{\em LowerTangent$(#1,#2)$}}
\newcommand{\fn}[1]{f(#1)}
\newcommand{\Ln}[1]{L(#1)}
\newcommand{\bt}{\beta}

\newcommand{\LC}[1]{{\em left$(#1)$}}
\newcommand{\RC}[1]{{\em right$(#1)$}}
\newcommand{\Kn}[1]{K#1}
\newcommand{\e}[2]{e_{#1}{(#2)}}
\newcommand{\pmp}[1]{pmp(#1)}

\title{Packing Plane Perfect Matchings into a Point Set\thanks{Research supported by NSERC.}}
\author{
Ahmad Biniaz\thanks{School of Computer Science, Carleton University, 
                    Ottawa, Canada.}
\and 
Prosenjit Bose\footnotemark[2]
\and
Anil Maheshwari\footnotemark[2]
\and 
Michiel Smid\footnotemark[2]
}
\date{\today}
\newtheorem{lemma}{Lemma}
\newtheorem{corollary}{Corollary}

\newtheorem{theorem}{Theorem}
\newtheorem{observation}{Observation}
\begin{document}

\maketitle

\begin{abstract}
Given a set $P$ of $n$ points in the plane, where $n$ is even, we consider the following question: How many plane perfect matchings can be packed into $P$? We prove that at least $\lceil\log_2{n}\rceil-2$ plane perfect matchings can be packed into any point set $P$. For some special configurations of point sets, we give the exact answer. We also consider some extensions of this problem.
\end{abstract}
\section{Introduction}
\label{introduction-section}
Let $P$ be a set of $n$ points in general position (no three points on a line) in the plane. A {\em geometric graph} $G=(P,E)$ is a graph whose vertex set is $P$ and whose edge set $E$ is a set of straight-line segments with endpoints in $P$. We say that two edges of $G$ {\em cross} each other if they have a point in common that is interior to both edges. Two edges are {\em disjoint} if they have no point in common. A subgraph $S$ of $G$ is said to be {\em plane} ({\em non-crossing} or {\em crossing-free}) if its edges do not cross. A {\em plane matching} is a plane graph consisting of pairwise disjoint edges. Two subgraphs $S_1$ and $S_2$ are {\em edge-disjoint} if they do not share any edge. A {\em complete geometric graph} $\Kn{(P)}$ is a geometric graph on $P$ which contains a straight-line edge between every pair of points in $P$. 

We say that the sequence $S_1,S_2,S_3,\dots$ of subgraphs is {\em packed into} $\Kn{(P)}$, if the subgraphs in this sequence are pairwise edge-disjoint. In a packing problem, we ask for the largest number of subgraphs of a given type that can be packed into $\Kn{(P)}$. Among all subgraphs of $\Kn{(P)}$, {\em plane perfect matchings}, {\em plane spanning trees}, and {\em plane spanning paths} are of interest. That is, one may look for the maximum number of plane spanning trees, plane Hamiltonian paths, or plane perfect matchings that can be packed into $\Kn{(P)}$. Since $\Kn{(P)}$ has $\frac{n(n-1)}{2}$ edges, at most $\frac{n}{2}$ spanning trees, at most $\frac{n}{2}$ spanning paths, and at most $n-1$ perfect matchings can be packed into it. 

A long-standing open question is to determine if the edges of $\Kn{(P)}$ (where $n$ is even) can be partitioned into $\frac{n}{2}$ plane spanning trees. In other words, is it possible to pack $\frac{n}{2}$ plane spanning trees into $\Kn{(P)}$? If $P$ is in convex position, the answer in the affirmative follows from the result of Bernhart and Kanien~\cite{Bernhart1979}. For $P$ in general position, Aichholzer et al.~\cite{Aichholzer2014} prove that $\Omega(\sqrt{n})$ plane spanning trees can be packed into $\Kn{(P)}$. They also show the existence of at least 2 edge-disjoint plane spanning paths. 

In this paper we consider a closely related question: How many plane perfect matchings can be packed into $\Kn{(P)}$, where $P$ is a set of $n$ points in general position in the plane, with $n$ even? 
\subsection{Previous Work}
\label{previous-work-section}
\subsubsection{Existence of Plane Subgraphs}
The existence of certain plane subgraphs in a geometric graph on a set $P$ of $n$ points is one of the classical problems in combinatorial and computational geometry. 

One of the extremal problems in geometric graphs which was first studied by Avital and Hanani~\cite{Avital1966}, Kuptiz~\cite{Kupitz1979}, Erd\H{o}s~\cite{Erdos1946}, and Perles (see reference \cite{Toth1999}) is the following. What is the smallest number $\e{k}{n}$ such that any geometric graph with $n$ vertices and more than $\e{k}{n}$ edges contains $k+1$ pairwise disjoint edges, i.e., a plane matching of size at least $k+1$. Note that $k\le\lfloor n/2\rfloor-1$. By a result of Hopf and Pannwitz~\cite{Hopf1934} and Erd\H{o}s~\cite{Erdos1946}, $\e{1}{n}=n$, i.e., any geometric graph with $n+1$ edges contains a pair of disjoint edges. Specifically, they showed that the diameter of a point set, i.e., the maximum distance among $n$ points in the plane, can be repeated $n$ times. Let $G_{max}(P)$ be the geometric graph on $P$ obtained by connecting each pair of points with the maximum distance by an edge. If $G_{max}(P)$ contains two disjoint edges $pq$ and $rs$, then the convex hull of $\{p,q,r,s\}$ forms either a triangle or a quadrilateral. In both cases there is a distance longer than $|pq|$ and $|rs|$; which is a contradiction. Thus $G_{max}(P)$ cannot have two disjoint edges. From this, it follows that $\e{1}{n}=n$.

Alon and Erd\H{o}s~\cite{Alon1989} proved that $\e{2}{n}<6n-5$, i.e., any geometric graph with $n$ vertices and at least $6n-5$ edges contains a plane matching of size three. This bound was improved to $\e{2}{n}\le 3n$ by Goddard et al.~\cite{Goddard1996}. Recently {\v{C}}ern{\'{y}}~\cite{Cerny2005} proved that $\e{2}{n}\le \lfloor 2.5 n\rfloor$; while the lower bound of $\e{2}{n}\ge \lceil2.5n\rceil-3$ is due to Perles (see \cite{Cerny2005}). For $\e{3}{n}$, Goddard et al.~\cite{Goddard1996} showed that $3.5n-6\le\e{3}{n}\le 10n$, which was improved by T{\'{o}}th and Valtr~\cite{Toth1999} to $4n-9\le\e{3}{n}\le 8.5n$.

For general values of $k$, Akiyama and Alon~\cite{Akiyama1989} gave the upper bound of $\e{k}{n}=O(n^{2-1/(k+1)})$. Goddard et al.~\cite{Goddard1996} improved the bound to $\e{k}{n}=O(n(\log n)^{k-3})$. Pach and T{\"{o}}r{\H{o}}csik~\cite{Pach1994} obtained the upper bound of $\e{k}{n} \le k^4n$; which is the first upper bound that is linear in $n$. The upper bound was improved to $k^3(n+1)$ by T{\'{o}}th and Valtr~\cite{Toth1999}; they also gave the lower bound of $\e{k}{n}\ge \frac{3}{2}(k-1)n-2k^2$. T{\'{o}}th~\cite{Toth2000} improved the upper bound to $\e{k}{n}\le 2^9k^2n$, where the constant has been improved to $2^8$ by Felsner~\cite{Felsner2004}. It is conjectured that $\e{k}{n}\le ckn$ for some constant $c$.

For the maximum value of $k$, i.e., $k=\frac{n}{2}-1$, with $n$ even, Aichholzer et al.~\cite{Aichholzer2010} showed that $\e{n/2-1}{n}={n \choose 2}-\frac{n}{2}=\frac{n(n-2)}{2}$. That is, by removing $\frac{n}{2}-1$ edges from any complete geometric graph, the resulting graph has $k+1=\frac{n}{2}$ disjoint edges, i.e., a plane perfect matching. This bound is tight; there exist complete geometric graphs, such that by removing $\frac{n}{2}$ edges, the resulting graph does not have any plane perfect matching. Similar bounds were obtained by Kupitz and Perles for complete convex graphs, i.e., complete graphs of point sets in convex position. 
Specifically, for convex geometric graphs, Kupitz and Perles showed that $\e{k}{n}=kn$; see \cite{Goddard1996} (see also \cite{Akiyama1989} and \cite{Alon1989}). In particular, in the convex case, $2n + 1$ edges guarantee a plane matching of size three. In addition, Keller and Perles~\cite{Keller2012} gave a characterization of all sets of $\frac{n}{2}$ edges whose removal prevents the resulting graph from having a plane perfect matching.

\v{C}ern{\'{y}} et al.~\cite{Cerny2007} considered the existence of Hamiltonian paths in geometric graphs. They showed that after removing at most $\sqrt{n}/(2\sqrt{2})$ edges from any complete geometric graph of $n$ vertices, the resulting graph still contains a plane Hamiltonian path. Aichholzer et al.~\cite{Aichholzer2010} obtained tight bounds on the maximum number of edges that can be removed from a complete geometric graph, such that the resulting graph contains a certain plane subgraph; they considered plane perfect matchings, plane subtrees of a given size, and triangulations. 
\subsubsection{Counting Plane Graphs}
The number of plane graphs of a given type in a set of $n$ points is also of interest. In 1980, Newborn and Moser~\cite{Newborn1980} asked for the maximal number of plane Hamiltonian cycles; they give an upper bound of $2\cdot6^{n-2}\lfloor\frac{n}{2}\rfloor!$, but conjecture that it should be of the form $c^n$, for some constant $c$. In 1982, Ajtai et al.~\cite{Ajtai1982} proved that the number of plane graphs is at most $10^{13n}$. Every plane graph is a subgraph of some triangulation (with at most $3n-6$ edges). Since a triangulation has at most $2^{3n-6}$ plane subgraphs, as noted in~\cite{Garcia2000}, any bound of $\alpha^n$ on the number of triangulations implies a bound of $2^{3n-6}\alpha^n<(8\alpha)^n$ on the number of plane graphs. The best known upper bound of $59^n$, for the number of triangulations is due to Santos and Seidel~\cite{Santos2003}. This implies the bound $472^n$ for plane graphs. As for plane perfect matchings, since a perfect matching has $\frac{n}{2}$ edges, Dumitrescu~\cite{Dumitrescu1999} obtained an upper bound of ${{3n-6} \choose {n/2}}\alpha^n\le (3.87\alpha)^n$, where $\alpha=59$. Recently, Sharir and Welzl~\cite{Sharir2006} showed that the number of plane perfect matchings is at most $O(10.05^n)$. They also showed that the number of all (not necessarily perfect) plane matchings is at most $O(10.43^n)$. 

Garc{\'{\i}}a et al.~\cite{Garcia2000} showed that the number of plane perfect matchings of a fixed size set of points in the plane is minimum when the points are in convex position. Motzkin~\cite{Motzkin1948} showed that points in convex position have $C_{n/2}$ many perfect matchings (classically referred to as non-crossing configurations of chords on a circle), where $C_{n/2}$ is the $(n/2)^{th}$ {\em Catalan number}; $C_{n/2}=\Theta(n^{-3/2}2^n)$. Thus, the number of plane perfect matchings of $n$ points in the plane is at least $C_{n/2}$. Garc{\'{\i}}a et al.~\cite{Garcia2000} presented a configuration of $n$ points in the plane which has $\Omega(n^{-4}3^n)$ many plane perfect matchings. See Table~\ref{table1}.

\subsubsection{Counting Edge-Disjoint Plane Graphs}

The number of edge-disjoint plane graphs of a given type in a point set $P$ of $n$ points is also of interest. Nash-Williams~\cite{Nash-Williams1961} and Tutte~\cite{Tutte1961} independently considered the number of (not necessarily plane) spanning trees. They obtained necessary and sufficient conditions for a graph to have $k$ edge-disjoint spanning trees. Kundu~\cite{Kundu1974} showed that any $k$-edge-connected graph contains at least $\lceil\frac{k-1}{2}\rceil$ edge-disjoint spanning trees. 

As for the plane spanning trees a long-standing open question is to determine if the edges of $\Kn{(P)}$ (where $n$ is even) can be partitioned into $\frac{n}{2}$ plane spanning trees. In other words, is it possible to pack $\frac{n}{2}$ plane spanning trees into $\Kn{(P)}$? If $P$ is in convex position, the answer in the affirmative follows from the result of Bernhart and Kanien~\cite{Bernhart1979}. In \cite{Bose2006}, the authors characterize the partitions of the complete convex graph into plane spanning trees. They also describe a sufficient condition, which generalizes the convex case, for points in general position. Aichholzer et al.~\cite{Aichholzer2014} showed that if the convex hull of $P$ contains $h$ vertices, then $\Kn{(P)}$ contains at least $\lfloor\frac{h}{2}\rfloor$ edge-disjoint plane spanning trees, and if $P$ is in a ``regular wheel configuration'', $\Kn{(P)}$ can be partitioned into $\frac{n}{2}$ spanning trees. For $P$ in general position they showed that $\Kn{(P)}$ contains $\Omega(\sqrt{n})$ edge-disjoint plane spanning trees. They obtained the following trade-off between the number of edge-disjoint plane spanning trees and the maximum vertex degree in each tree: For any $k\le \sqrt{n/12}$, $\Kn{(P)}$ has $k$ edge-disjoint plane spanning trees with maximum vertex degree $O(k^2)$ and diameter $O(\log(n/k^2))$. They also showed the existence of at least 2 edge-disjoint plane Hamiltonian paths. 

\subsection{Our Results}
\label{our-results-section}

Given a set $P$ of $n$ points in the plane, with $n$ even, we consider the problem of packing plane perfect matchings into $\Kn{(P)}$. 
From now on, a {\em matching} will be a {\em perfect matching}. 

In Section~\ref{edge-disjoint-plane-section} we prove bounds on the number of plane matchings that can be packed into $\Kn{(P)}$. 
In Section~\ref{convex-position-section} we show that if $P$ is in convex position, then $\frac{n}{2}$ plane matchings can be packed into $\Kn{(P)}$; this bound is tight. 

The points in wheel configurations are considered in Section~\ref{wheel-section}. We show that if $P$ is in regular wheel configuration, then $\frac{n}{2}-1$ edge-disjoint plane matchings can be packed into $\Kn{(P)}$; this bound is tight as well. In addition, for a fixed size set of points, we give a wheel configuration of the points which contains at most $\lceil\frac{n}{3}\rceil$ edge-disjoint plane matchings. 

Point sets in general position are considered in Section~\ref{general-position-section}. We show how to find three edge-disjoint plane matchings in any set of at least 8 points. If $n$ is a power of two, we prove that $\Kn{(P)}$ contains at least $\log_2n$ many edge-disjoint plane matchings. For the general case, where $n$ is an even number, we prove that $\Kn{(P)}$ contains at least $\lceil\log_2n\rceil-2$ edge-disjoint plane matchings. 

In Section~\ref{non-crossing-matching-section} we count the number of pairwise non-crossing plane matchings. Two plane matchings $M_1$ and $M_2$ are called {\em non-crossing} if the edges of $M_1$ and $M_2$ do not cross each other. We show that $\Kn{(P)}$ contains at least two and at most five non-crossing plane matchings; these bounds are tight. Table~\ref{table1} summarizes the results. 

In Section~\ref{persistency-section} we define the concept of {\em matching persistency} in a graph. A graph $G$ is called {\em matching-persistent}, if by removing any perfect matching $M$ from $G$, the resulting graph, $G-M$, still contains a perfect matching. We define the {\em plane matching persistency} of a point set $P$, denoted by $\pmp{P}$, to be the smallest number of edge-disjoint plane matchings such that, if we remove them from $\Kn{(P)}$ the resulting graph does not have any plane perfect matching. In other words, $\pmp{P}=|\mathcal{M}|$, where $\mathcal{M}$ is the smallest set of edge-disjoint plane matchings such that $\Kn{(P)}-\bigcup_{M\in\mathcal{M}}{M}$ does not have any plane perfect matching. Here, the challenge is to find point sets with high plane matching persistency. We show that $\pmp{P}\ge 2$ for all point sets $P$. We give a configuration of $P$ with $\pmp{P}\ge 3$. 
Concluding remarks and open problems are presented in Section~\ref{conclusion}.

\begin{table}
\caption{Number of plane perfect matchings in a point set $P$ of $n$ points ($n$ is even).}
\label{table1}
\centering
    \begin{tabular}{|l||@{\hskip 0.15in}c@{\hskip 0.15in}|@{\hskip 0.15in}c@{\hskip 0.15in}|@{\hskip 0.15in}c@{\hskip 0.15in}|@{\hskip 0.15in}c@{\hskip 0.15in}|}
         \hline
             Matching 	& $\forall P: \ge$ &$\exists P:\le$&$\exists P: \ge$&$\forall P:\le$  \\ \hline\hline
             total& 	$2^n$\cite{Garcia2000, Motzkin1948}&$2^n$\cite{Motzkin1948}&$3^n$\cite{Garcia2000}& $O(10.05^n)$\cite{Sharir2006}\\\hline\hline
		 edge-disjoint&$\lceil\log_2{n}\rceil-2$ &$\lceil \frac{n}{3}\rceil$& $\frac{n}{2}$&$n-1$\\
             non-crossing edge-disjoint& 2 & 2&5&5 \\
         \hline
    \end{tabular}
\end{table}

\section{Preliminaries}
\label{preliminaries}

\subsection{Graph-Theoretical Background}
\label{graph-background-section}
Consider a graph $G=(V,E)$ with vertex set $V$ and edge set $E$. If $G$ is a complete graph on a vertex set $V$ of size $n$, then $G$ is denoted by $K_n$.
A $k${\em -factor} is a regular graph of degree $k$. If $G$ is the union of pairwise edge-disjoint $k$-factors, their union is called a $k${\em-factorization} and $G$ itself is $k${\em-factorable}~\cite{Harary1991}. A {\em matching} in a graph $G$ is a set of edges that do not share vertices. A {\em perfect matching} is a matching which matches all the vertices of $G$. Since a perfect matching is a regular graph of degree one, it is a $1$-factor. It is well-known that for $n$ even, the complete graph $K_n$ is 1-factorable (See~\cite{Harary1991}).
Note that $K_n$ has $\frac{n(n-1)}{2}$ edges and every $1$-factor has $\frac{n}{2}$ edges. Thus, $K_n$ can be partitioned into at most $n-1$ edge-disjoint perfect matchings.

On the other hand it is well-known that the edges of a complete graph $K_n$, where $n$ is even, can be colored by $n − 1$ colors such that any two disjoint edges have a different color. Each color is assigned to $\frac{n}{2}$ edges, so that each color defines a $1$-factor. The following geometric construction of a coloring, which uses a ``regular wheel configuration'', is provided in \cite{Soifer2009}. In a regular wheel configuration, $n − 1$ regularly spaced points are placed on a circle and one point is placed at the center of the circle. For each color class, include an edge $e$ from the center to one of the boundary vertices, and all of the edges perpendicular to the line through $e$, connecting pairs of boundary vertices.

The number of perfect matchings in a complete graph $K_n$ (with $n$ even), denoted by $M(n)$, is given by the double factorial; $M(n)=(n − 1)!!$ \cite{Callan2009}, where $(n-1)!!=\allowbreak (n-1)\times\allowbreak  (n-3)\times\allowbreak \dots\times 3\times 1$. We give the following recurrence for the number of perfect matchings in $K_n$. Consider an edge $e\in K_n$. The number of perfect matchings in $K_n$ is the number of perfect matchings containing $e$, plus the number of perfect matchings which do not contain $e$. One can easily derive the following recurrence for $M(n)$: 
\[ M(n) = \left\{
  \begin{array}{l l}
    1 & \quad \text{if $n=2$}\\
    3 & \quad \text{if $n=4$}\\
  M(n-2)+2{n-2 \choose 2}M(n-4) & \quad \text{if $n\ge 6$}
  \end{array} \right.\]

\subsection{Plane Matchings in Colored Point Sets}
\label{colored-matching-section}
Let $P$ be a set of $n$ colored points in general position in the plane with $n$ even. A {\em colored matching} of $P$, is a perfect matching such that every edge connects two points of distinct colors. A {\em plane colored matching} is a colored matching which is non-crossing.  
A special case of a plane colored matching, where $P$ is partitioned into a set $R$ of $\frac{n}{2}$ red points and a set $B$ of $\frac{n}{2}$ blue points, is called {\em plane bichromatic matching}, also known as {\em red-blue matching} ($RB${\em-matching}). In other words, an $RB$-matching of $P$ is a non-crossing perfect matching such that every edge connects a red point to a blue point. It is well-known that if no three points of $P$ are collinear, then $P$ has an $RB$-matching~\cite{Putnam1979}. As shown in Figure~\ref{RB-fig}(a), some point sets have a unique $RB$-matching. Hershberger and Suri~\cite{Hershberger1992} construct an $RB$-matching in $O(n\log n)$ time, which is optimal.

\begin{figure}[htb]
  \centering
\setlength{\tabcolsep}{0in}
  $\begin{tabular}{cc}
\multicolumn{1}{m{.5\columnwidth}}{\centering\includegraphics[width=.33\columnwidth]{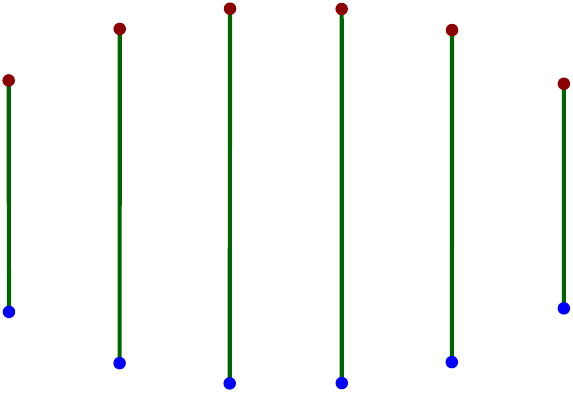}}
&\multicolumn{1}{m{.5\columnwidth}}{\centering\includegraphics[width=.35\columnwidth]{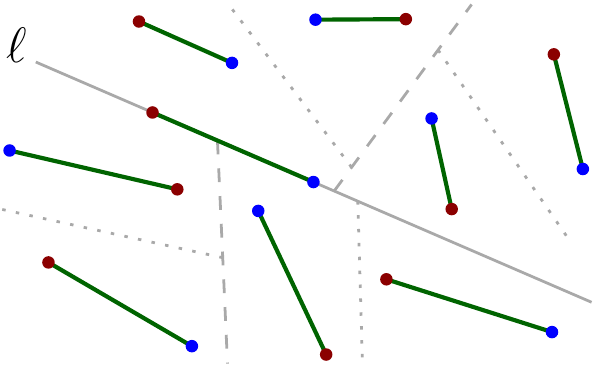}} \\
(a) & (b)
\end{tabular}$
  \caption{(a) A point set with a unique $RB$-matching, (b) Recursive ham sandwich cuts: first cut is in solid, second-level cuts are in dashed, and third-level cuts are in dotted lines.}
\label{RB-fig}
\end{figure}

We review some proofs for the existence of a plane perfect matching between $R$ and $B$:
\begin{itemize}
\item {\Min{R}{B}:} Consider a matching $M$ between $R$ and $B$ which minimizes the total Euclidean length of the edges. The matching $M$ is plane. To prove this, suppose that two line segments $r_1b_1$ and $r_2b_2$ in $M$ intersect. By the triangle inequality, $|r_1b_2|+|r_2b_1|<|r_1b_1|+|r_2b_2|$. This implies that by replacing $r_1b_1$ and $r_2b_2$ in $M$ by $r_1b_2$ and $r_2b_1$, the total length of the matching is decreased; which is a contradiction.

\item {\Cut{R}{B}:} The {\em ham sandwich theorem} implies that there is a line $\ell$, known as a {\em ham sandwich cut}, that splits both $R$ and $B$ exactly in half; if the size of $R$ and $B$ is odd, the line passes through one of each. Match the two points on $\ell$ (if there are any) and recursively solve the problem on both sides of $\ell$; the recursion stops when each subset has one red point and one blue point. By matching these two points in all subsets, a plane perfect matching for $P$ is obtained. See Figure~\ref{RB-fig}(b). A ham sandwich cut can be computed in $O(n)$ time \cite{Lo1994}, and hence the running time can be expressed as the recurrence $T(n)=O(n)+2\cdot T(\lfloor\frac{n}{2}\rfloor)$. Therefore, an $RB$-matching can be computed in $O(n\log n)$ time. 

\item {\Tangent{R}{B}:} If $R$ and $B$ are separated by a line, we can compute an $RB$-matching in the following way. W.l.o.g. assume that $R$ and $B$ are separated by a vertical line $\ell$. Let \CH{R} and \CH{B} denote the convex hulls of $R$ and $B$. Compute the upper tangent $rb$ of \CH{R} and \CH{B} where $r\in R$ and $b\in B$. Match $r$ and $b$, and recursively solve the problem for $R-\{r\}$ and $B-\{b\}$; the recursion stops when the two subsets are empty. In each iteration, all the remaining points are below the line passing through $r$ and $b$, thus, the line segments representing a matched pair in the successor iterations do not cross $rb$. Therefore, the resulting matching is plane.
\end{itemize}

Consider a set $P$ of $n$ points where $n$ is even, and a partition $\{P_1,\dots,P_k\}$ of $P$ into $k$ color classes. Sufficient and necessary conditions for the existence of a colored matching in $P$ follows from the following theorem by Sitton~\cite{Sitton1996}:

\begin{theorem}[Sitton~\cite{Sitton1996}]
\label{Sitton}
Let $K_{n_1,\dots,n_k}$ be a complete multipartite graph with $n$ vertices, where $n_1\le\dots\le n_k$. If $n_k\allowbreak \le \allowbreak n_1+\dots+\allowbreak n_{k-1}$, then $K_{n_1,\dots,n_k}$ has a matching of size $\lfloor\frac{n}{2}\rfloor$. 
\end{theorem}

Aichholzer et al.~\cite{Aichholzer2010} showed that if $K_{n_1,\dots,n_k}$ is a geometric graph corresponding to a colored point set $P$, then the minimum-weight colored matching of $P$ is non-crossing. Specifically, they extend the proof of 2-colored point sets to multi-colored point sets:

\begin{theorem}[Aichholzer et al.~\cite{Aichholzer2010}]
\label{Aichholzer}
Let $P$ be a set of colored points in general position in the plane with $|P|$ even. Then $P$
admits a non-crossing perfect matching such that every edge connects two points of distinct colors if and only if at most half the points in $P$ have the same color.
\end{theorem}

\section{Packing Plane Matchings into Point Sets}
\label{edge-disjoint-plane-section}
Let $P$ be a set of $n$ points in the plane with $n$ even. In this section we prove lower bounds on the number of plane matchings that can be packed into $\Kn{(P)}$. It is obvious that every point set has at least one plane matching, because a minimum weight perfect matching in $\Kn{(P)}$, denoted by ${\sf Min}(P)$, is plane. A trivial lower bound of 2 (for $n\ge4$) is obtained from a minimum weight Hamiltonian cycle in $\Kn{(P)}$, because this cycle is plane and consists of two edge-disjoint matchings. We consider points in convex position (Section~\ref{convex-position-section}), wheel configuration (Section~\ref{wheel-section}), and general position (Section~\ref{general-position-section}).

\subsection{Points in Convex Position}
\label{convex-position-section}
In this section we consider points in convex position. We show that if $P$ is in convex position, $\frac{n}{2}$ plane matchings can be packed into $\Kn{(P)}$; this bound is tight.
\begin{lemma}
\label{two-convex-edges}
 If $P$ is in convex position, where $|P|$ is even and $|P|\ge4$, then every plane matching in $P$ contains at least two edges of \CH{P}.
\end{lemma}
\begin{proof}
Let $M$ be a plane matching in $P$. We prove this lemma by induction on the size of $P$. If $|P|=4$, then $|M|=2$. None of the diagonals of $P$ can be in $M$, thus, the two edges in $M$ belong to \CH{P}. If $|P|>4$ then $|M|\ge 3$. If all edges of $M$ are edges of \CH{P}, then the claim in the lemma holds. Assume that $M$ contains a diagonal edge $pq\notin \CH{P}$. Let $P_1$ and $P_2$ be the sets of points of $P$ on each side of $\ell(p,q)$ (both including $p$ and $q$). Let $M_1$ and $M_2$ be the edges of $M$ in $P_1$ and $P_2$, respectively. It is obvious that $P_1$ (resp. $P_2$) is in convex position and $M_1$ (resp. $M_2$) is a plane matching in $P_1$ (resp. $P_2$). By the induction hypothesis $M_1$ (resp. $M_2$) contains two edges of \CH{P_1} (resp. \CH{P_2}). Since $\CH{P}=\CH{P_1}\cup \CH{P_2}$ and $|M_1\cap M_2|=1$, $M$ contains at least two edges of \CH{P}.
\end{proof}

\begin{theorem}
\label{convex}
For any set $P$ of $n$ points in convex position in the plane, with $n$ even, the maximum number of plane matchings that can be packed into $\Kn{(P)}$ is $\frac{n}{2}$.
\end{theorem}
\begin{proof}
By Lemma~\ref{two-convex-edges}, every plane matching in $P$ contains at least two edges of \CH{P}. On the other hand, \CH{P} has $n$ edges. Therefore, the number of plane matchings that can be packed into $\Kn{(P)}$ is at most $\frac{n}{2}$.

Now we show how to pack $\frac{n}{2}$ plane matchings into $\Kn{(P)}$.
Let $P=\{p_0,\dots,p_{n-1}\}$, and w.l.o.g. assume that $p_0,p_1,\dots,p_{n-1}$ is the radial ordering of the points in $P$ with respect to a fixed point in the interior of \CH{P}. For each $p_i$ in the radial ordering, where $0\le i<\frac{n}{2}$, let $M_i=\{p_{i+j-1}p_{n+i-j}:j=1,\dots,\frac{n}{2}\}$ (all indices are modulo $n$). Informally speaking, $M_i$ is a plane perfect matching obtained from edge $p_ip_{i-1}$ and all edges parallel to $p_ip_{i-1}$; see Figure~\ref{convex-fig}. Let $\mathcal{M}=\{M_i:i=0,\dots,\frac{n}{2}-1\}$. The matchings in $\mathcal{M}$ are plane and pairwise edge-disjoint. Thus, $\mathcal{M}$ is a set of $\frac{n}{2}$ plane matchings that can be packed into $\Kn{(P)}$.  
\end{proof}

\begin{figure}[htb]
  \centering
\setlength{\tabcolsep}{0in}
\includegraphics[width=.32\columnwidth]{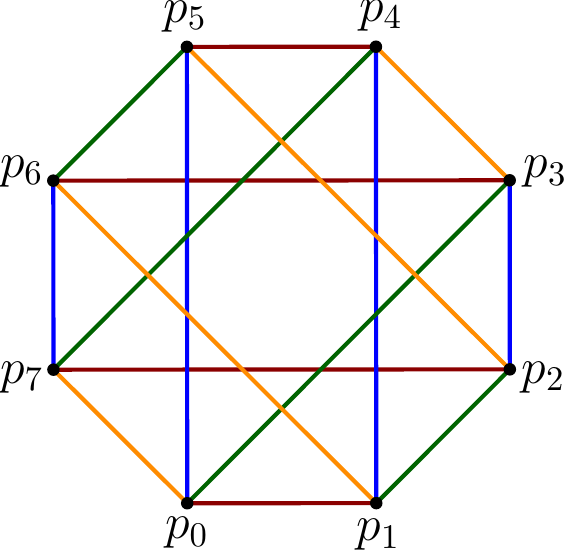}
  \caption{Points in convex position.}
\label{convex-fig}
\end{figure}

\subsection{Points in Wheel Configurations}
\label{wheel-section}
A point set $P$ of $n$ points is said to be in ``regular wheel configuration'' in the plane, if $n-1$ points of $P$ are regularly spaced on a circle $C$ and one point of $P$ is at the center of $C$.
We introduce a variation of the regular wheel configuration as follows. 
Let the point set $P$ be partitioned into $X$ and $Y$ such that $|X|\ge 3$ and $|X|$ is an odd number. The points in $X$ are regularly spaced on a circle $C$. For any two distinct points $p, q\in X$ let $\ell(p,q)$ be the line passing through $p$ and $q$. Since $X$ is regularly spaced on $C$ and $|X|$ is an odd number, $\ell(p,q)$ does not contain the center of $C$. Let $H(p,q)$ and $H'(p,q)$ be the two half planes defined by $\ell(p,q)$ such that $H'(p,q)$ contains the center of $C$. Let $C'=\bigcap_{p,q\in X}H'(p,q)$. The points in $Y$ are in the interior of $C'$; see Figure~\ref{n-over-3-fig}(a). Note that for any two points $p,q\in X$, the line segment $pq$ does not intersect the interior of $C'$. The special case when $|Y|=1$ is the regular wheel configuration.

\begin{figure}[htb]
  \centering
\setlength{\tabcolsep}{0in}
  $\begin{tabular}{ccc}
\multicolumn{1}{m{.5\columnwidth}}{\centering\includegraphics[width=.4\columnwidth]{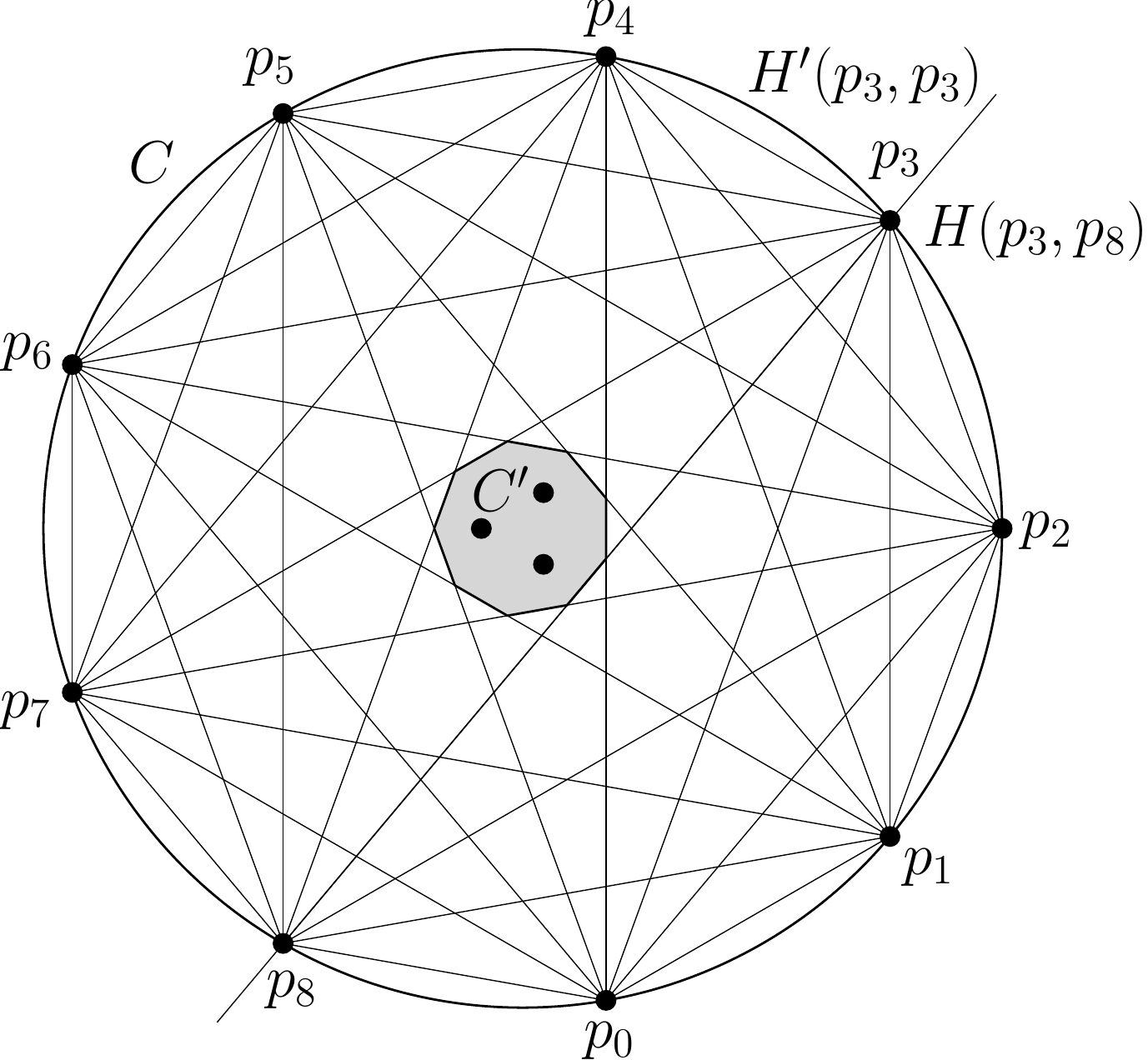}}&
\multicolumn{1}{m{.5\columnwidth}}{\centering\includegraphics[width=.4\columnwidth]{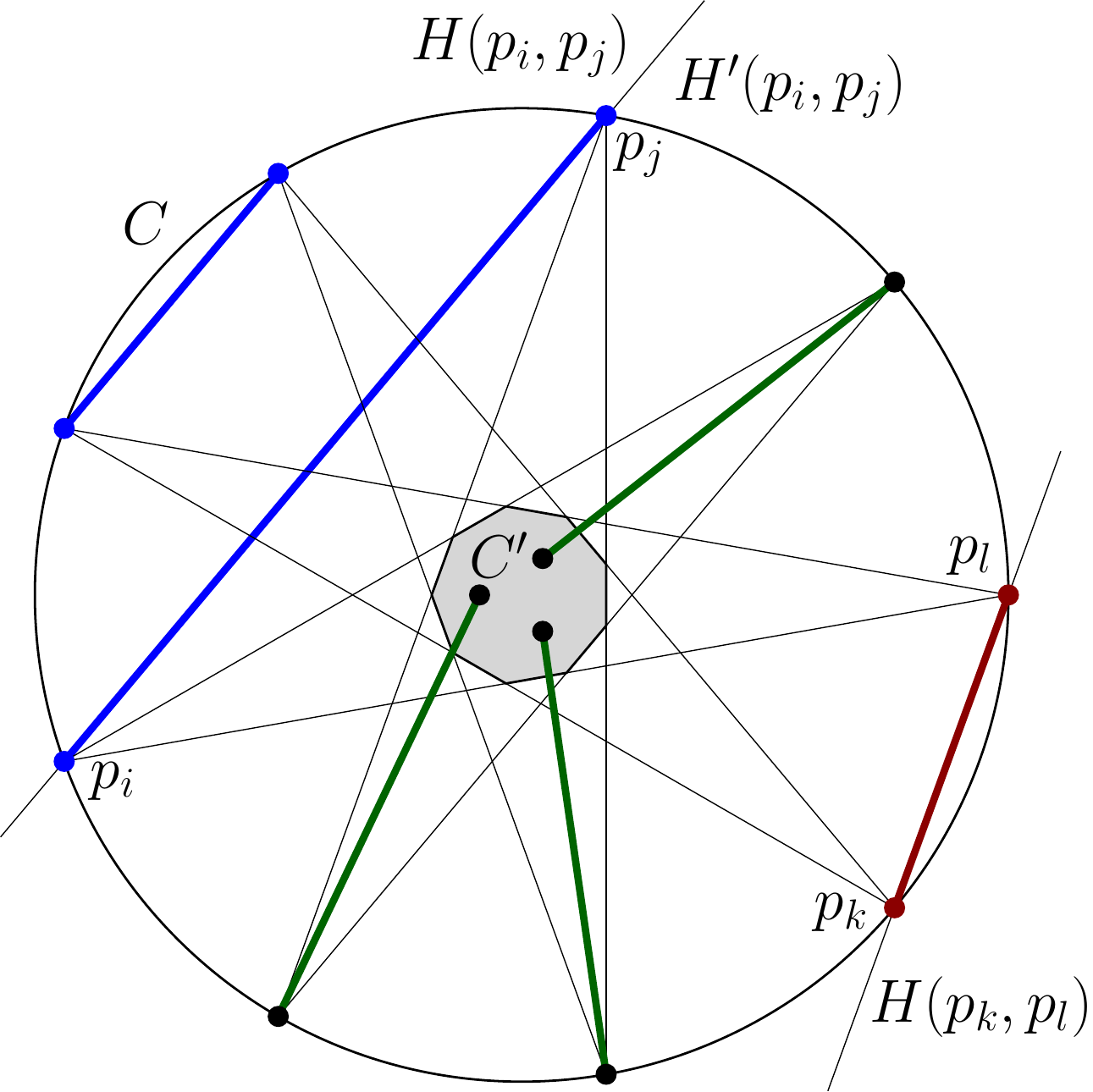}}
\\
(a) & (b)
\end{tabular}$
  \caption{(a) A variation of the regular wheel configuration. (b) Illustration of Theorem~\ref{n-over-3-thr}. The points of $A$ and the edges of $M(A)$ are in blue, and the points of $B$ and the edges of $M(B)$ are in red.}
\label{n-over-3-fig}
\end{figure}

\begin{lemma}
\label{two-wheel-edges}
 Let $P$ be a set of points in the plane where $|P|$ is an even number and $|P|\ge 6$. Let $\{X,Y\}$ be a partition of the points in $P$  such that $|X|$ is an odd number and $|Y|\le 2\lfloor\frac{|P|}{6}\rfloor-1$. If $P$ is in the wheel configuration described above, then any plane matching in $P$ contains at least two edges of \CH{P}.
\end{lemma}
\begin{proof}
 Consider a plane matching $M$ of $P$. It is obvious that $\CH{P}=\CH{X}$; we show that $M$ contains at least two edges of $\CH{X}$. Note that $|X|=|P|-|Y|$, and both $|X|$ and $|Y|$ are odd numbers. Observe that $|X|\ge 4\lfloor\frac{|P|}{6}\rfloor+1=2|Y|+3\ge 5$; which implies that $|Y|\le \frac{|X|-1}{2}-1$. Thus, $|X|>|Y|$, and hence there is at least one edge in $M$ with both endpoints in $X$. Let $p_ip_j$ be the longest such edge. Recall that $C'\subset H'(p_i,p_j)$. Let $A$ be the set of points of $P$ in $H(p_i,p_j)$ (including $p_i$ and $p_j$), and let $A'$ be the set of points of $P$ in $H'(p_i,p_j)$ (excluding $p_i$ and $p_j$). By definition, $H(p_i,p_j)\cup \ell(p_i,p_j)$ does not contain any point of $Y$. Thus, $A\subset X$ and $A$ is in convex position with $|A|\le \frac{|X|-1}{2}$ (note that $|X|$ is an odd number). Let $M(A)$ and $M(A')$ be the edges of $M$ induced by the points in $A$ and $A'$, respectively. Clearly, $\{M(A),M(A')\}$ is a partition of the edges of $M$, and hence $M(A)$ (resp. $M(A')$) is a plane perfect matching for $A$ (resp. $A'$). We show that each of $M(A)$ and $M(A')$ contains at least one edge of \CH{X}. First we consider $M(A)$. If $|A|=2$, then $p_ip_j$ is the only edge in $M(A)$ and it is an edge of \CH{X}. Assume that $|A|\ge 4$. By Lemma~\ref{two-convex-edges}, $M(A)$ contains at least two edges of \CH{A}. On the other hand each edge of \CH{A}, except for $p_ip_j$, is also an edge of \CH{X}; see Figure~\ref{n-over-3-fig}(b). This implies that $M(A)-\{p_ip_j\}$ contains at least one edge of \CH{X}. Now we consider $M(A')$. Let $X'=A'\cap X$, that is, $\{A,X'\}$ is a partition of the points in $X$. Since $|A|\le \frac{|X|-1}{2}$, we have $|X'|\ge\frac{|X|+1}{2}$. Recall that $|Y|\le \frac{|X|-1}{2}-1$. Thus, $|Y|<|X'|$, and hence there is an edge $p_kp_l\in M(A')$ with both $p_k$ and $p_l$ in $X'$. Let $B$ be the set of points of $P$ in $H(p_k,p_l)$ (including $p_k$ and $p_l$). By definition, $H(p_k,p_l)\cup \ell(p_k,p_l)$ does not contain any point of $Y$. Thus, $B\subset X$ and $B$ is in convex position. On the other hand, by the choice of $p_ip_j$ as the longest edge, $A$ cannot be a subset of $B$ and hence $B\subset X'$. Let $M(B)$  be the edges of $M(A')$ induced by the points in $B$. We show that $M(B)$ contains at least one edge of \CH{X}. If $|B|=2$, then $p_kp_l$ is the only edge in $M(B)$ and it is an edge of \CH{X}. Assume that $|B|\ge 4$. By Lemma~\ref{two-convex-edges}, $M(B)$ contains at least two edges of \CH{B}. On the other hand, each edge of \CH{B}, except for $p_kp_l$, is also an edge of \CH{X}; see Figure~\ref{n-over-3-fig}(b). This implies that $M(B)-\{p_kp_l\}$ contains at least one edge of \CH{X}. This completes the proof. 
\end{proof}

\begin{figure}[htb]
  \centering
\setlength{\tabcolsep}{0in}
  $\begin{tabular}{cc}

\multicolumn{1}{m{.5\columnwidth}}{\centering\includegraphics[width=.33\columnwidth]{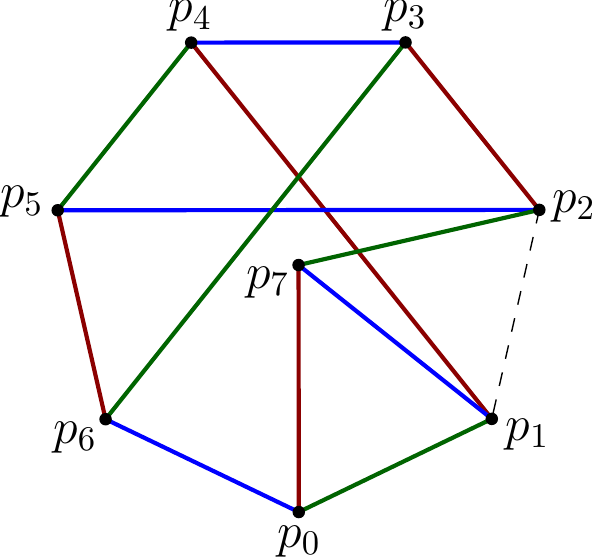}}
&\multicolumn{1}{m{.5\columnwidth}}{\centering\includegraphics[width=.32\columnwidth]{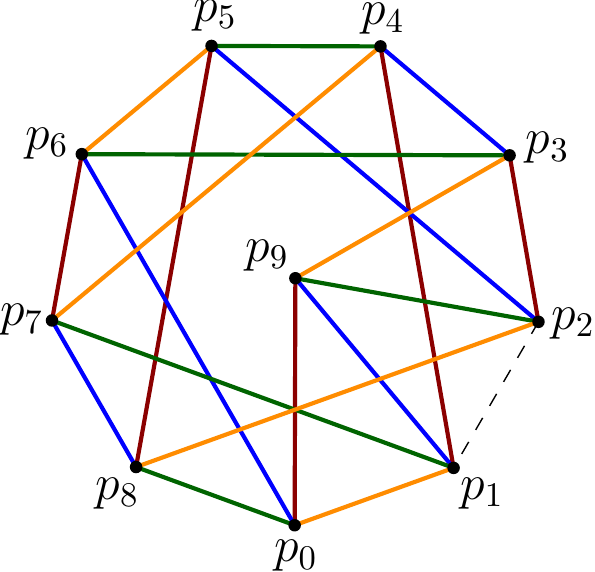}} \\
(a) & (b)
\end{tabular}$
  \caption{Points in the regular configuration with (a) $n=4k$ and (b) $n=4k+2$; one of the edges in \CH{P} cannot be matched.}
\label{wheel-fig}
\end{figure}

\begin{theorem}
\label{wheel}
For a set $P$ of $n\ge 6$ points in the regular wheel configuration in the plane with $n$ even, the maximum number of plane matchings that can be packed into $\Kn{(P)}$ is $\frac{n}{2}-1$.
\end{theorem}
\begin{proof}
In the regular wheel configuration, $P$ is partitioned into a point set $X$ of size $n-1$ and a point set $Y$ of size 1. The points of $X$ are placed regularly on a circle $C$ and the (only) point of $Y$ is the center of $C$. By Lemma~\ref{two-wheel-edges}, every plane matching in $P$ contains at least two edges of \CH{P}. On the other hand, \CH{P} has $n-1$ edges. Therefore, the number of plane matchings that can be packed into $\Kn{(P)}$ is at most $\frac{n-1}{2}$. Since $n$ is an even number and the number of plane matchings is an integer, we can pack at most $\frac{n}{2}-1$ plane matchings into $\Kn{(P)}$.

Now we show how to pack $\frac{n}{2}-1$ plane matchings into $\Kn{(P)}$.
Let $P=\{p_0,\dots,p_{n-1}\}$, and w.l.o.g. assume that $p_{n-1}$ is the center of $C$. Let $P'=P-\{p_{n-1}\}$, and let $p_0,p_1,\dots,p_{n-2}$ be the radial ordering of the points in $P'$ with respect to $p_{n-1}$. For each $p_i$ in the radial ordering, where $0\le i\le \frac{n}{2}-2$, let $$R_i=\{p_{i+j}p_{i+2\lceil(n-2)/4\rceil-j+1}:j=1,\dots,\lceil(n-2)/4\rceil\},$$ and $$L_i=\{p_{i-j}p_{i-2\lfloor(n-2)/4\rfloor+j-1}:j=1,\dots,\lfloor(n-2)/4\rfloor\}$$ (all indices are modulo $n-1$). Let $M_i=R_i\cup L_i\cup \{p_ip_{n-1}\}$; informally speaking, $M_i$ is a plane perfect matching obtained from edge $p_ip_{n-1}$ and edges parallel to $p_ip_{n-1}$. See Figure~\ref{wheel-fig}(a) for the case where $n=4k$ and Figure~\ref{wheel-fig}(b) for the case where $n=4k+2$. Let $\mathcal{M}=\{M_i:i=0,\dots,\frac{n}{2}-2\}$. The matchings in $\mathcal{M}$ are plane and pairwise edge-disjoint. Thus, $\mathcal{M}$ is a set of $\frac{n}{2}-1$ plane matchings that can be packed into $\Kn{(P)}$. 
\end{proof}

In the following theorem we use the wheel configuration to show that for any even integer $n\ge 6$, there exists a set $P$ of $n$ points in the plane, such that no more than $\lceil\frac{n}{3}\rceil$ plane matchings can be packed into $\Kn{(P)}$. 

\begin{theorem}
\label{n-over-3-thr}
For any even number $n\ge 6$, there exists a set $P$ of $n$ points in the plane such that no more than $\lceil\frac{n}{3}\rceil$ plane matchings can be packed into $\Kn{(P)}$.
\end{theorem}
\begin{proof}
The set $P$ of $n$ points is partitioned into $X$ and $Y$, where $|Y|=2\lfloor\frac{n}{6}\rfloor-1$ and $|X|=n-|Y|$. The points in $X$ are regularly placed on a circle $C$ and the points in $Y$ are in the interior $\bigcap_{p,q\in X}H'(p,q)$. By Lemma~\ref{two-wheel-edges}, any plane matching in $P$ contains at least two edges of \CH{P}.
Since $\CH{P}=\CH{X}$, any plane matching of $P$ contains at least two edges of \CH{X}. Thus, if $\mathcal{M}$ denotes any set of plane matchings which can be packed into $\Kn{(P)}$, we have (note that $|X|$ is odd)
$$|\mathcal{M}|\le \frac{|X|-1}{2}=\frac{n-2\lfloor n/6\rfloor}{2}=\frac{n}{2}-\lfloor\frac{n}{6}\rfloor\le \frac{n}{2}-\frac{n-5}{6}\le \lceil\frac{n}{3}\rceil.$$
\end{proof}

\subsection{Points in General Position}
\label{general-position-section}
In this section we consider the problem of packing plane matchings for point sets in general position (no three points on a line) in the plane. 
Let $P$ be a set of $n$ points in general position in the plane, with $n$ even. Let $M(P)$ denote the maximum number of plane matchings that can be packed into $\Kn{(P)}$. As mentioned earlier, a trivial lower bound of $2$ (when $n \ge 4$) is obtained from a minimum weight Hamiltonian cycle, which is plane and consists of two edge-disjoint perfect matchings. 

In this section we show that at least $\lceil\log_2{n}\rceil-2$ plane matchings can be packed into $\Kn{(P)}$. As a warm-up, we first show that if $n$ is a power of two, then $\log_2{n}$ plane matchings can be packed into $\Kn{(P)}$. Then we extend this result to get a lower bound of $\lceil\log_2{n}\rceil-2$ for every point set with an even number of points. We also show that if $n\ge 8$, then at least three plane matchings can be packed into $\Kn{(P)}$, which improves the result for $n=10$, $12$, and $14$. Note that, as a result of Theorem~\ref{n-over-3-thr}, there exists a set of $n=6$ points such that no more than $\lceil\frac{n}{3}\rceil=2$ plane matchings can be packed into $\Kn{(P)}$. First consider the following observation.

\begin{observation}
\label{partition-obs}
 Let $\mathcal{P}=\{P_1,\dots, P_k\}$ be a partition of the point set $P$, such that $|P_i|$ is even and $\CH{P_i}\cap\CH{P_j}=\emptyset$ for all $1\le i,j\le k$ where $i\neq j$. Let $i$ be an index such that, $M(P_i)=\min\{M(P_j):1\le j\le k\}$. Then, $M(P)\ge M(P_i)$.
\end{observation}

\begin{theorem}
\label{n-power2}
For a set $P$ of $n$ points in general position in the plane, where $n$ is a power of 2, at least $\log_2{n}$ plane matchings can be packed into $\Kn{(P)}$.
\end{theorem}
\begin{proof}
We prove this theorem by induction. The statement of the theorem holds for the base case, where $n=2$. Assume that $n\ge 4$. Recall that $M(P)$ denotes the maximum number of plane matchings that can be packed into $\Kn{(P)}$. W.l.o.g. assume that a vertical line $\ell$ partitions $P$ into sets $R$ and $B$, each of size $\frac{n}{2}$. By the induction hypothesis, $M(R),M(B)\ge\log_2{(\frac{n}{2})}$. By Observation~\ref{partition-obs}, $M(P)\ge \min\{M(R),M(B)\}\ge\log_2{(\frac{n}{2})}$. That is, by pairing a matching $M_R$ in $R$ with a matching $M_B$ in $B$ we get a plane matching $M_P$ in $\Kn{(P)}$, such that each edge in $M_P$ has both endpoints in $R$ or in $B$. If we consider the points in $R$ as red and the points in $B$ as blue, \Cut{R}{B} (see Section~\ref{colored-matching-section}) gives us a plane perfect matching $M'_P$ in $\Kn{(P)}$, such that each edge in $M'_P$ has one endpoint in $R$ and one endpoint in $B$. That is $M'_P\cap M_P=\emptyset$. Therefore, we obtain one more plane matching in $\Kn{(P)}$, which implies that $M(P)\ge \log_2{(\frac{n}{2})}+1=\log_2n$.
\end{proof}

Let $R$ and $B$ be two point sets which are separated by a line. A {\em crossing tangent} between $R$ and $B$ is a line $l$ touching \CH{R} and \CH{B} such that $R$ and $B$ lie on different sides of $l$. Note that $l$ contains a point $r\in R$, a point $b\in B$, and consequently the line segment $rb$; we say that $l$ is subtended from $rb$. It is obvious that there are two (intersecting) crossing tangents between $R$ and $B$; see Figure~\ref{three-matching-fig}. 

\begin{figure}[htb]
  \centering
\setlength{\tabcolsep}{0in}
  $\begin{tabular}{cc}
\multicolumn{1}{m{.5\columnwidth}}{\centering\includegraphics[width=.35\columnwidth]{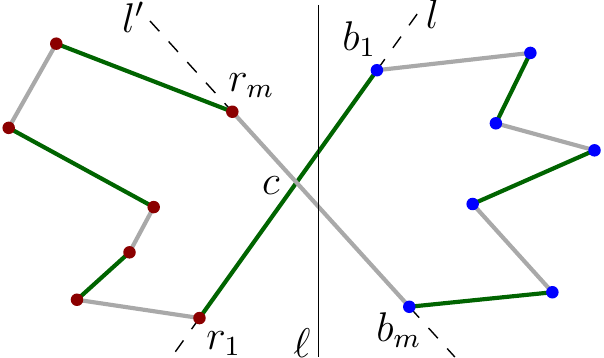}}
&\multicolumn{1}{m{.5\columnwidth}}{\centering\includegraphics[width=.35\columnwidth]{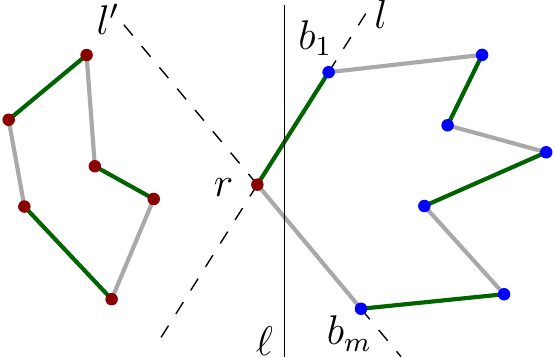}} \\
(a) & (b)
\end{tabular}$
  \caption{(a) The crossing tangents intersect at a point $c\notin P$: $R$ and $B$ are sorted clockwise around $c$, (b) The crossing tangents intersect at a point $r\in R$: $B$ is sorted clockwise around $r$. $M_1$ and $M_2$ are shown by green and gray line segments.}
\label{three-matching-fig}
\end{figure}

\begin{theorem}
\label{3-matching-theorem}
For a set $P$ of $n\ge 8$ points in general position in the plane with $n$ even, at least three plane matchings can be packed into $\Kn{(P)}$.
\end{theorem}
\begin{proof}
We describe how to extract three edge-disjoint plane matchings, $M_1,M_2,M_3$, from $\Kn{(P)}$. Let $\ell$ be a vertical line which splits $P$ into sets $R$ and $B$, each of size $\frac{n}{2}$. Consider the points in $R$ as red and the points in $B$ as blue. We differentiate between two cases: (a) $n=4k$ and (b) $n= 4k+2$, for some integer $k>1$.

In case (a), both $R$ and $B$ have an even number of points. Let \M{1}{R} and \M{2}{R} (resp. \M{1}{B} and \M{2}{B}) be two edge-disjoint plane matchings in $R$ (resp. $B$) obtained by a minimum length Hamiltonian cycle in $R$ (resp. $B$). Let $M_1=\M{1}{R}\cup \M{1}{B}$ and $M_2=\M{2}{R}\cup \M{2}{B}$. Clearly $M_1$ and $M_2$ are edge-disjoint plane matchings for $P$. Let $M_3=\Cut{R}{B}$. It is obvious that $M_3$ is edge-disjoint from $M_1$ and $M_2$, which completes the proof in the first case.

In case (b), both $R$ and $B$ have an odd number of points and we cannot get a perfect matching in each of them. Let $l$ and $l'$ be the two crossing tangents between $R$ and $B$, subtended from $rb$ and $r'b'$, respectively. We differentiate between two cases: (i) $l$ and $l'$ intersect in the interior of $rb$ and $r'b'$, (ii) $l$ and $l'$ intersect at an endpoint of both $rb$ and $r'b'$; see Figure~\ref{three-matching-fig}.
\begin{itemize}
\item In case (i), let $c$ be the intersection point; see Figure~\ref{three-matching-fig}(a). Let $r_1, r_2, \dots, r_m$ and $b_1, b_2, \allowbreak\dots, \allowbreak b_m$ be the points of $R$ and $B$, respectively, sorted clockwise around $c$, where $m=\frac{n}{2}$, $r_1=r,r_m=r', b_1=b, b_m=b'$. Consider the Hamiltonian cycle $H=\{r_ir_{i+1}:1\le i< m\}\cup \{b_ib_{i+1}:1\le i< m\}\cup \{r_1b_1,r_mb_m\}$. Let $M_1$ and $M_2$ be the two edge-disjoint matchings obtained from $H$. Note that $r_1b_1$ and $r_mb_m$ cannot be in the same matching, thus, $M_1$ and $M_2$ are plane. Let $M_3=\Tangent{R}{B}$. As described in Section~\ref{colored-matching-section}, $M_3$ is a plane matching for $P$. In order to prove that $M_3\cap (M_1\cup M_2)=\emptyset$, we show that $rb$ and $r'b'$\textemdash which are the only edges in $M_1\cup M_2$ that connect a point in $R$ to a point in $B$\textemdash do not belong to $M_3$. Note that \Tangent{R}{B} iteratively selects an edge which has the same number of red and blue points below its supporting line, whereas the supporting lines of $rb$ and $r'b'$ have different numbers of red and blue points below them. Thus $rb$ and $r'b'$ are not considered by \Tangent{R}{B}. Therefore $M_3$ is edge-disjoint from $M_1$ and $M_2$. 

\item In case (ii), w.l.o.g. assume that $l$ and $l'$ intersect at the red endpoint of $rb$ and $r'b'$, i.e., $r=r'$; See Figure~\ref{three-matching-fig}(b). Let $R'=R\setminus \{r\}$ and $B'=B\cup \{r\}$. Note that both $R'$ and $B'$ have an even number of points and $|R'|,|B'|\ge 4$. Let \M{1}{R'} and \M{2}{R'} be two edge-disjoint plane matchings in $R'$ obtained by a minimum length Hamiltonian cycle in $R'$. Let $b_1, b_2, \dots, b_m$ be the points of $B$ sorted clockwise around $r$, where $m=\frac{n}{2}$, $b_1=b, b_m=b'$. Consider the Hamiltonian cycle $\HC{R'}=\{b_ib_{i+1}:1\le i< m\}\cup \{rb_1,rb_m\}$. Let $M_1(B')$ and $M_2(B')$ be the two edge-disjoint plane matchings in $B'$ obtained from $\HC{B'}$. Let $M_1=\M{1}{R'}\cup \M{1}{B'}$ and $M_2=\M{2}{R'}\cup \M{2}{B'}$. Clearly $M_1$ and $M_2$ are edge-disjoint plane matchings in $P$. Let $M_3=\Tangent{R}{B}$. As described in case (i), $M_3$ is a plane matching in $P$ and $M_3\cap (M_1\cup M_2)=\emptyset$. Therefore, $M_3$ is edge-disjoint from $M_1$ and $M_2$.
\end{itemize}
\end{proof}
As a direct consequence of Theorem~\ref{n-power2} and Theorem~\ref{3-matching-theorem} we have the following corollary.
\begin{corollary}
For a set $P$ of $n=2^i\cdot m$ points in general position in the plane with $m\ge 4$, at least $i+2$ plane matchings can be packed into $\Kn{(P)}$.
\end{corollary}
\begin{proof}
 Partition $P$ by vertical lines, into $2^{i-1}$ point sets, each of size $2m$. By Theorem~\ref{3-matching-theorem}, at least three plane matchings can be packed into each set. Considering these sets as the base cases in Theorem~\ref{n-power2}, we obtain $i-1$ plane matchings between these sets. Thus, in total, $i+2$ plane matchings can be packed into $\Kn{(P)}$
\end{proof}

\begin{theorem}
For a set $P$ of $n$ points in general position in the plane, with $n$ even, at least $\lceil\log_2{n}\rceil-2$ plane matchings can be packed into $\Kn{(P)}$.
\end{theorem}
\begin{proof}
We describe how to pack a set $\mathcal{M}$ of $\lceil\log_2{n}\rceil-2$ plane perfect matchings into $\Kn{(P)}$. The construction consists of the following three main steps which we will describe in detail.

\begin{enumerate}
  \item Building a binary tree $T$.
  \item Assigning the points of $P$ to the leaves of $T$.
  \item Extracting $\mathcal{M}$ from $P$ using internal nodes of $T$.
\end{enumerate}

\begin{paragraph}{\em \small 1. Building the tree T.}
In this step we build a binary tree $T$ such that each node of $T$ stores an even number, and each internal node of $T$ has a left and a right child. For an internal node $u$, let \LC{u} and \RC{u} denote the left child and the right child of $u$, respectively. Given an even number $n$, we build $T$ in the following way:
\begin{itemize}
  \item The root of $T$ stores $n$.
  \item If a node of $T$ stores $2$, then that node is a leaf.
  \item For a node $u$ storing $m$, with $m$ even and $m\ge 4$, we store the following even numbers into \LC{u} and \RC{u}:
    \begin{itemize}
	\item If $m$ is divisible by $4$, we store $\frac{m}{2}$ in both \LC{u} and \RC{u}; see Figure~\ref{tree-construction-fig}(a).
	\item If $m$ is not divisible by $4$ and $u$ is the root or the left child of its parent then we store $2\lfloor\frac{m}{4}\rfloor$ in \LC{u} and $m-2\lfloor\frac{m}{4}\rfloor$ in \RC{u}; see Figure~\ref{tree-construction-fig}(b).
	\item If $m$ is not divisible by $4$ and $u$ is the right child of its parent then we store $m-2\lfloor\frac{m}{4}\rfloor$ in \LC{u} and $2\lfloor\frac{m}{4}\rfloor$ in \RC{u}; see Figure~\ref{tree-construction-fig}(c).
    \end{itemize}
Note that in the last two cases\textemdash where $m$ is not divisible by four\textemdash the absolute difference between the values stored in \LC{u} and \RC{u} is exactly 2. See Figure~\ref{matching-example-fig}.
\end{itemize}
\end{paragraph}

\begin{figure}[htb]
  \centering
\setlength{\tabcolsep}{0in}
  $\begin{tabular}{ccc}
\multicolumn{1}{m{.33\columnwidth}}{\centering\includegraphics[width=.16\columnwidth]{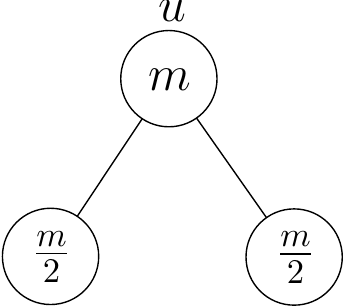}}&
\multicolumn{1}{m{.33\columnwidth}}{\centering\includegraphics[width=.2\columnwidth]{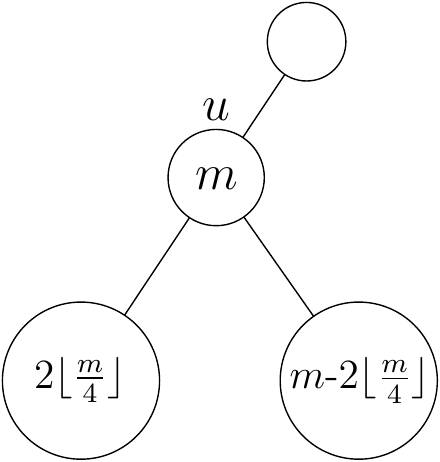}}
&\multicolumn{1}{m{.33\columnwidth}}{\centering\includegraphics[width=.2\columnwidth]{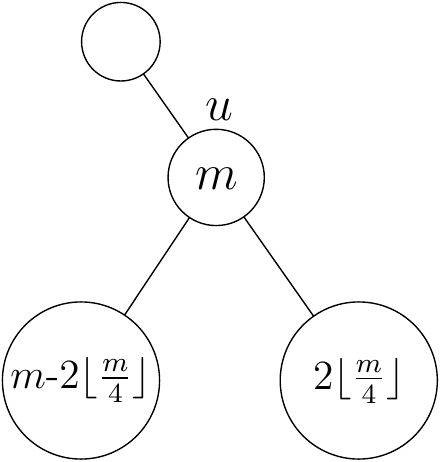}} \\
(a) & (b)&(c)
\end{tabular}$
  \caption{(a) $m$ is divisible by four, (b) $m$ is not divisible by four and $u$ is a left child, and (c) $m$ is not divisible by four and $u$ is a right child.}
\label{tree-construction-fig}
\end{figure}

\begin{paragraph}{\em \small 2. Assigning the points to the leaves of the tree.}
In this step we describe how to assign the points of $P$, in pairs, to the leaves of $T$. We may assume without loss of generality that no two points of $P$ have the same $x$-coordinate. Sort the points of $P$ in a increasing order of their $x$-coordinate.
Assign the first two points to the leftmost leaf, the next two points to the second leftmost leaf, and so on. Note that $T$ has $\frac{n}{2}$ leaves, and hence all the points of $P$ are assigned to the leaves of $T$. See Figure~\ref{matching-example-fig}.  
\end{paragraph}

\begin{paragraph}{\em \small 3. Extracting the matchings.}
Let $L$ be the number of edges in a shortest path from the root to any leaf in $T$; in Figure~\ref{matching-example-fig}, $L=3$. For an internal node $u\in T$, let $T_u$ be the subtree rooted at $u$. Let $L_u$ and $R_u$ be the set of points assigned to the left and right subtrees of $T_u$, respectively, and let $P_u=L_u\cup R_u$. Consider the points in $L_u$ as red and the points in $R_u$ as blue. Since the points in $L_u$ have smaller $x$-coordinates than the points in $R_u$, we say that $L_u$ and $R_u$ are separated by a vertical line $\ell(u)$. For each internal node $u$ where $u$ is in level $0\le i< L$ in $T$\textemdash assuming the root is in level $0$\textemdash we construct a plane perfect matching $M_u$ in $P_u$ in the following way. Let $m$ be the even number stored at $u$.
\begin{itemize}
  \item If $m$ is divisible by $4$ (Figure~\ref{tree-construction-fig}(a)), then let $M_u=\Min{L_u}{R_u}$; see Section~\ref{colored-matching-section}. Since $|L_u|=|R_u|$, $M_u$ is a plane perfect matching for $P_u$. See vertices $u_2, u_3$ in Figure~\ref{matching-example-fig}.

  \item If $m$ is not divisible by $4$ and $u$ is the root or a left child (Figure~\ref{tree-construction-fig}(b)), then $|R_u|-|L_u|=2$. Let $a,b$ be the two points assigned to the rightmost leaf in $T_u$, and let $M_u=\{ab\}\cup\Min{L_u}{R_u-\{a,b\}}$. Since $|L_u|=|R_u-\{a,b\}|$, $M_u$ is a perfect matching in $P_u$. In addition, $a$ and $b$ are the two rightmost points in $P_u$, thus, $ab$ does not intersect any edge in $\Min{L_u}{R_u-\{a,b\}}$, and hence $M_u$ is plane. See vertices $u_0, u_1, u_5$ in Figure~\ref{matching-example-fig}.

  \item If $m$ is not divisible by $4$ and $u$ is a right child (Figure~\ref{tree-construction-fig}(c)), then $|L_u|-|R_u|=2$. Let $a,b$ be the two points assigned to the leftmost leaf in $T_u$ and let $M_u=\{ab\}\cup\Min{L_u-\{a,b\}}{R_u}$. Since $|L_u-\{a,b\}|=|R_u|$, $M_u$ is a perfect matching in $P_u$. In addition, $a$ and $b$ are the two leftmost points in $P_u$, thus, $ab$ does not intersect any edge in $\Min{L_u-\{a,b\}}{R_u}$, and hence $M_u$ is plane. See vertices $u_4, u_6$ in Figure~\ref{matching-example-fig}.
\end{itemize}

For each $i$, where $0\le i< L$, let $S(i)$ be the set of vertices of $T$ in level $i$; see Figure~\ref{matching-example-fig}. For each level $i$ let $M_i=\bigcup_{u\in S(i)}{M_u}$. Let $\mathcal{M}=\{M_i:0\le i< L\}$. 
In the rest of the proof, we show that $\mathcal{M}$ contains $\lceil\log_2{n}\rceil-2$ edge-disjoint plane matchings in $P$.
\end{paragraph}

\begin{figure}[htb]
  \centering
  \includegraphics[width=.9\columnwidth]{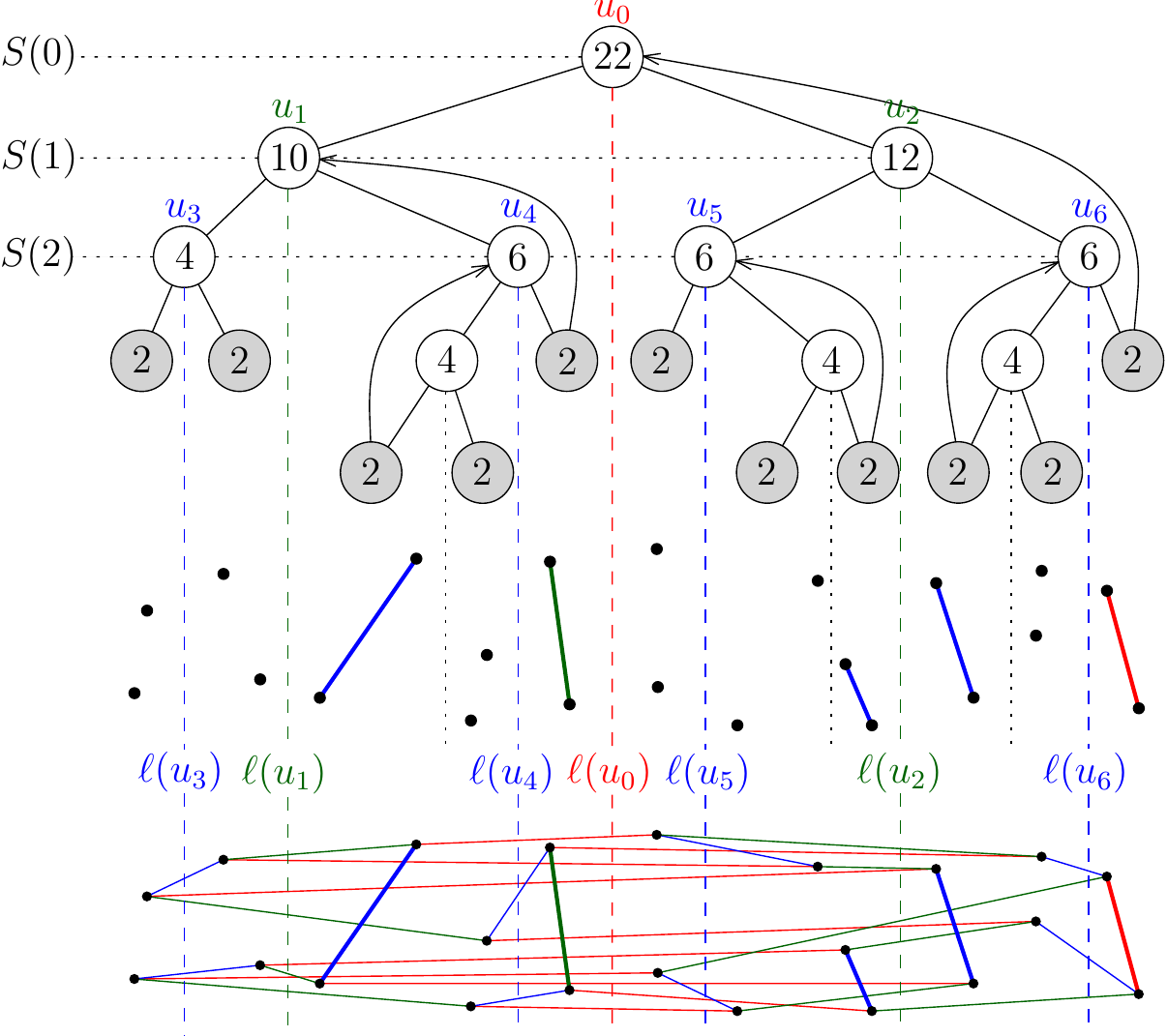}
 \caption{The points in $P$ are assigned, in pairs, to the leaves of $T$, from left to right. The point set $P$ with the edge-disjoint plane matchings is shown as well. $M_0$ contains the bold red edge and the red edges crossing $\ell(u_0)$. $M_1$ contains the bold green edge and the green edges crossing $\ell(u_1), \ell(u_2)$. $M_2$ contains the bold blue edges and the blue edges crossing $\ell(u_3), \allowbreak \ell(u_4), \allowbreak\ell(u_5), \allowbreak\ell(u_6)$.}
  \label{matching-example-fig}
\end{figure}

\vspace{10pt}
{\em {\bf Claim 1:} For each $i$, where $0\le i< L$, $M_i$ is a plane perfect matching in $P$.} 
Note that if $u$ is the root of the tree, then $P_u=P$. In addition, for each internal node $u$ (including the root), $\{L_u,R_u\}$ is a partition of the point set $P_u$. This implies that in each level $i$ of the tree, where $0\le i <L$, we have $P=\bigcup_{u\in S(i)}{P_u}$. Moreover, the points in $P$ are assigned to the leaves of $T$ in non-decreasing order of their $x$-coordinate. Thus, $\mathcal{P}_i=\{P_u: u\in S(i)\}$ is a partition of the point set $P$; the sets $P_u$ with $u\in S(i)$ are separated by vertical lines; see Figure~\ref{matching-example-fig}. Therefore, $M_i$ is a perfect plane matching in $P$; which proves the claim.
\vspace{10pt}

\vspace{10pt}
{\em {\bf Claim 2:} For all $M_i, M_j\in \mathcal{M}$, where $0\le i,j<L$ and $i\neq j$, $M_i\cap M_j=\emptyset$.}
In order to prove that $M_i$ and $M_j$ are edge-disjoint, we show that for each pair of distinct internal nodes $u$ and $v$, $M_u\cap M_v=\emptyset$. If $u$ and $v$ are in the same level, then $P_u$ and $P_v$ are separated by $\ell(u)$, thus, $M_u$ and $M_v$ do not share any edge. Thus, assume that $u\in S(i)$ and $v\in S(j)$ such that $0\le i, j< L$, $i\neq j$, and w.l.o.g. assume that $i< j$. If $v\notin T_u$, then $P_u$ and $P_v$ are separated by line $\ell(w)$, where $w$ is the lowest common ancestor of $u$ and $v$; this implies that $M_u$ and $M_v$ do not share any edge. Therefore, assume that $v\in T_u$, and w.l.o.g. assume that $v$ is in the left subtree of $T_u$. Thus, $P_v$\textemdash and consequently $M_v$\textemdash is to the left of $\ell(u)$. The case where $v$ is in the right subtree of $T_u$ is symmetric. Let $m\ge 4$ be the number stored at $u$. We differentiate between three cases:

\begin{itemize}
  \item If $m$ is divisible by $4$, then all the edges in $M_u$ cross $\ell(u)$, while the edges in $M_v$ are to the left of $\ell(u)$. This implies that $M_u$ and $M_v$ are disjoint. 

  \item If $m$ is not divisible by $4$ and $u$ is the root or a left child, then all the edges of $M_u$ cross $\ell(u)$, except the rightmost edge $ab$ which is to the right of $\ell(u)$. Since $M_v$ is to the left of $\ell(u)$, it follows that $M_u$ and $M_v$ are disjoint. 

  \item If $m$ is not divisible by $4$ and $u$ is a right child, then all the edges of $M_u$ cross $\ell(u)$, except the leftmost edge $ab$. If $a,b\notin P_v$, then $ab\notin M_v$, and hence $M_u$ and $M_v$ are disjoint. If $a,b\in P_v$ then $v$ is the left child of its parent and all the edges in $M_v$ cross $\ell(v)$ (possibly except one edge which is to the right of $\ell(v)$), while $ab$ is to the left of $\ell(v)$. Therefore $M_u$ and $M_v$ do not share any edge. This completes the proof of the claim.
\end{itemize}
\vspace{10pt}

\vspace{10pt}
{\em {\bf Claim 3:} For every two nodes $u$ and $v$ in the same level of $T$, storing $m$ and $m'$, respectively, $|m-m'| \le 2$.}

We prove the claim inductively for each level $i$ of $T$. For the base case, where $i=1$: (a) if $n$ is divisible by four, then both $u$ and $v$ store $\frac{n}{2}$ and the claim holds, (b) if $n$ is not divisible by four then $u$ stores $2\lfloor\frac{n}{4}\rfloor$ and $v$ stores $n-2\lfloor\frac{n}{4}\rfloor$; as $0\le n-2\lfloor\frac{n}{4}\rfloor-2\lfloor\frac{n}{4}\rfloor\le 2$, the claim holds for $i=1$. 
Now we show that if the claim is true for the $i$th level of $T$, then the claim is true for the $(i+1)$th level of $T$.
Let $u$ and $v$, storing $m$ and $m'$, respectively, be in the $i$th level of $T$. By the induction hypothesis, the claim holds for the $i$th level, i.e., $|m-m'|\le 2$. We prove that the claim holds for the $(i+1)$th level of $T$, i.e., for the children of $u$ and $v$. Since $m$ and $m'$ are even numbers, $|m-m'|\in\{0,2\}$. If $|m-m'|=0$, then $m=m'$, and by a similar argument as in the base case, the claim holds for the children of $u$ and $v$. If $|m-m'|=2$, then w.l.o.g. assume that $m'=m+2$. Let $\alpha$ be the smallest number and $\beta$ be the largest number stored at the children of $u$ and $v$ (which are at the $(i+1)$th level). We show that $\beta-\alpha \le 2$. It is obvious that $\alpha= 2\lfloor\frac{m}{4}\rfloor$ and $\beta= m'-2\lfloor\frac{m'}{4}\rfloor$. Thus,
\begin{align}
\label{alpha-beta}
\beta-\alpha & = m'-2\lfloor\frac{m'}{4}\rfloor-2\lfloor\frac{m}{4}\rfloor\nonumber\\
       & = m+2-2\lfloor\frac{m+2}{4}\rfloor-2\lfloor\frac{m}{4}\rfloor 
\end{align}
Now, we differentiate between two cases, where $m=4k$ or $m=4k+2$. If $m=4k$, then by Equation~\ref{alpha-beta},
\begin{align}
\beta-\alpha & = 4k+2-2\lfloor\frac{4k+2}{4}\rfloor-2\lfloor\frac{4k}{4}\rfloor\nonumber\\
       & = 4k+2-2k-2k\nonumber\\
	&= 2.\nonumber
\end{align}
If $m=4k+2$, then by Equation~\ref{alpha-beta},
\begin{align}
\beta-\alpha & = 4k+4-2\lfloor\frac{4k+4}{4}\rfloor-2\lfloor\frac{4k+2}{4}\rfloor\nonumber\\
       & = 4k+4-2(k+1)-2k\nonumber\\
	&= 2\nonumber
\end{align}
which completes the proof of the claim.
\vspace{10pt}

\vspace{10pt}
{\em {\bf Claim 4:} $L\ge\lceil\log_2{n}\rceil-2$.}

It follows from Claim 3 that all the leaves of $T$ are in the last two levels. Since $T$ has $\frac{n}{2}$ leaves, $T$ has $n-1$ nodes. Recall that $L$ is the number of edges in a shortest path from the root to any leaf in $T$. Thus, $L\ge h-1$, where $h$ is the height of $T$. To give a lower bound on $h$, one may assume that the last level of $T$ is also full, thus,

$$n-1\le 2^0+2^1+2^2+\dots+2^h\le 2^{h+1}-1$$
and hence, $h\ge \log_2{n} -1$. Therefore, $L\ge h-1\ge \log_2{n} -2$.  Since $L$ is an integer, $L\ge\lceil\log_2{n}\rceil-2$; which proves the claim.

\vspace{20pt}

Claim 1 and Claim 2 imply that $\mathcal{M}$ contains $L$ edge-disjoint plane perfect matchings. Claim 4 implies that $L\ge\lceil\log_2{n}\rceil-2$, which proves the statement of the theorem.
\end{proof}

\subsection{Non-crossing Plane Matchings}
\label{non-crossing-matching-section}
In this section we consider the problem of packing plane matchings into $\Kn{(P)}$ such that any two different matchings in the packing are non-crossing.
Two edge-disjoint plane matchings $M_1$ and $M_2$ are {\em non-crossing}, if no edge in $M_1$ crosses any edge in $M_2$, and vice versa. For a set $P$ of $n$ points in general position in the plane, with $n$ even, we give tight lower and upper bounds on the number of pairwise non-crossing plane perfect matchings that can be packed into $\Kn{(P)}$. 

\begin{lemma}
\label{5-non-crossing}
For a set $P$ of $n$ points in general position in the plane, with $n$ even, at most five pairwise non-crossing plane matchings can be packed into $\Kn{(P)}$.
\end{lemma}
\begin{proof}
 Let $\{M_1,M_2,\dots,M_m\}$ be any maximal set of non-crossing edge-disjoint plane matchings in $\Kn{(P)}$. Let $M=M_1\cup M_2\dots \cup M_m$, and let $G$ be the induced subgraph of $\Kn{(P)}$ by $M$. It is obvious that $G$ is an $m$-regular graph. Since $M_1,\dots,M_m$ are plane and pairwise non-crossing, $G$ is an $m$-regular plane graph. It is well known that every plane graph has a vertex of degree at most 5. Thus, $G$ has a vertex of degree at most five and hence $m\le 5$; which implies at most five pairwise non-crossing plane matchings can be packed into $\Kn{(P)}$.
\end{proof}

\begin{figure}[htb]
  \centering
  \includegraphics[width=.4\columnwidth]{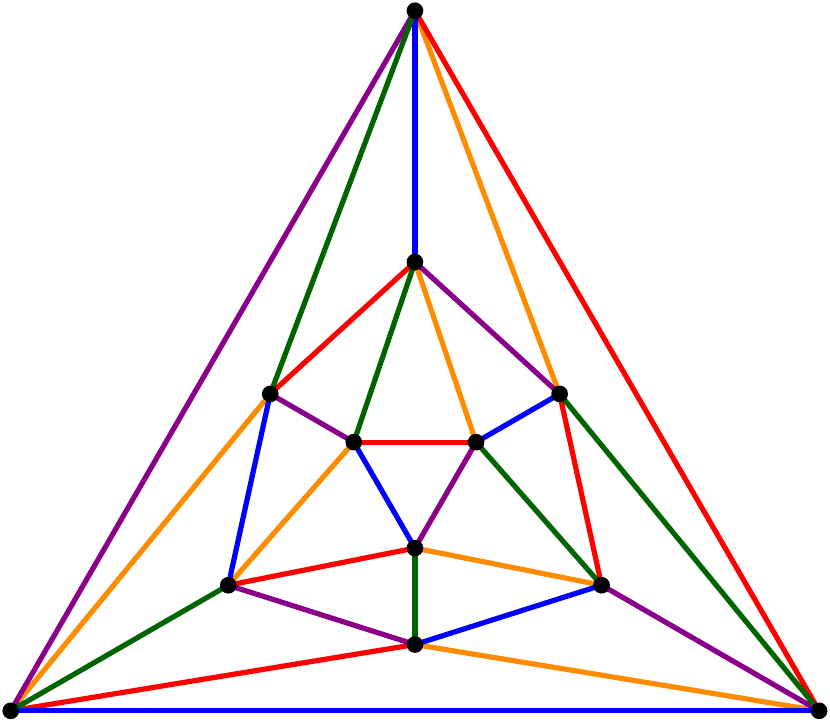}
 \caption{A point set with five non-crossing edge-disjoint plane perfect matchings.}
  \label{5-regular-fig}
\end{figure}

Figure~\ref{5-regular-fig} shows a $5$-regular geometric graph on a set of 12 points in the plane which contains five non-crossing edge-disjoint plane matchings. In \cite{Hasheminezhad2011}, the authors showed how to generate an infinite family of 5-regular planar graphs using the graph in Figure~\ref{5-regular-fig}. By an extension of the five matchings shown in Figure~\ref{5-regular-fig}, five non-crossing matchings for this family of graphs is obtained. Thus, the bound provided by Lemma~\ref{5-non-crossing} is tight.  

It is obvious that if $P$ contains at least four points, the minimum length Hamiltonian cycle in $\Kn{(P)}$ contains two non-crossing edge-disjoint plane matchings. In the following lemma we show that there exist point sets which contain at most two non-crossing edge-disjoint plane matchings.

\begin{observation}
\label{even-cycle-obs}
The union of two edge-disjoint perfect matchings in any graph is a set of even cycles.
\end{observation}

\begin{lemma}
 \label{2-non-crossing}
For a set $P$ of $n\ge 4$ points in convex position in the plane, with $n$ even, at most two pairwise non-crossing plane matchings can be packed into $\Kn{(P)}$.
\end{lemma}
\begin{proof}
 The proof is by contradiction. Consider three pairwise non-crossing plane matchings $M_1,M_2,M_3$. Let $\mathcal{C}$ be the union of $M_1$ and $M_2$. By Observation~\ref{even-cycle-obs}, $\mathcal{C}$ is a union of cycles, say $C_1,\dots, C_k$, each with an even number of points. Observe that $C_i$, where $1\le i\le k$, is plane and convex, and has at least 4 vertices. 
In addition, the regions enclosed by $C_1,C_2,\dots,C_k$ are pairwise disjoint.

If $k=1$, then $C_1=\CH{P}$. By Lemma~\ref{two-convex-edges}, $M_3$ contains two edges of $\CH{P}$, which contradicts that $M_1$, $M_2$, and $M_3$ are pairwise disjoint. 

From now on, we assume that $k\ge 2$. Then, each $C_i$ contains an edge that is not an edge of \CH{P}.

\vspace{10pt}
{\em {\bf Claim 1:} There exists a cycle $C_i$ such that at most one edge of $C_i$ is not an edge of \CH{P}.}

Let $pq$ be an edge of $C_1\cup C_2 \cup \dots\cup C_k$ such that $pq$ is not an edge of \CH{P} and $\min(|H(p,q)\cap P|,|H'(p,q)\cap P|)$ is minimum, where $H(p,q)$ and $H'(p,q)$ are the two half planes defined by a line through $p$ and $q$. The cycle containing $pq$ satisfies the statement in the claim.
\vspace{10pt}

Let $C_i$ be a cycle in $\mathcal{C}$ which satisfies the statement of Claim 1. Since $M_1$, $M_2$ and $M_3$ are pairwise non-crossing, none of the edges in $M_3$ intersect $pq$. Thus, in $M_3$ we have one of the following two cases:
\begin{itemize}
 \item Both $p$ and $q$ are matched to the points of $C_i$. Let $P'$ be the points of $C_i$. Note that $|P'|$ is an even number and $|P'|\ge 4$. In addition $C_i=\CH{P'}$. Let $M'_3$ be the edges in $M_3$ which are induced by the points in $P'$. By Lemma~\ref{two-convex-edges}, $M'_3$ contains two edges of \CH{P'}. Since $C_i=\CH{P'}$, both of these edges belong to $C_i$. This contradicts the fact that $M_1$, $M_2$, $M_3$ are pairwise edge-disjoint.
 \item None of $p$ and $q$ are matched to the points of $C_i$. Let $P'$ be the points of $C_i$ except $p$ and $q$. Note that $|P'|$ is an even number. Let $M'_3$ be the edges in $M_3$ which are induced by the points in $P'$. If $|P'|=2$, then the only edge in $M'_3$ is an edge of $C_i$. If $|P'|\ge 4$, then by Lemma~\ref{two-convex-edges}, $M'_3$ contains two edges of \CH{P'}. At least one of these edges belong to $C_i$. Either case contradicts the fact that $M_1$, $M_2$, $M_3$ are pairwise edge-disjoint.
\end{itemize}
\end{proof}
We conclude this section with the following theorem.
\begin{theorem}
For a set $P$ of $n\ge 4$ points in general position in the plane, with $n$ even, at least two and at most five pairwise non-crossing plane matchings can be packed into $\Kn{(P)}$. These bounds are tight.
\end{theorem}

\section{Matching Removal Persistency}
\label{persistency-section}
In this section we define the matching persistency of a graph. A graph $G$ is {\em matching persistent} if by removing any perfect matching $M$ from $G$, the resulting graph, $G-M$, has a perfect matching. We define the {\em matching persistency} of $G$, denoted by $mp(G)$, as the size of the smallest set $\mathcal{M}$ of edge-disjoint perfect matchings that can be removed from $G$ such that $G-\mathcal{M}$ does not have any perfect matching. 
In other words, if $mp(G)=k$, then

\begin{enumerate}
 \item by removing an arbitrary set of $k-1$ edge-disjoint perfect matchings from $G$, the resulting graph still contains a perfect matching, and 
\item there exists a set of $k$ edge-disjoint perfect matchings such that by removing these matchings from $G$, the resulting graph does not have any perfect matching.
\end{enumerate}

In particular, $G$ is matching persistent iff $mp(G)\ge 2$.

\begin{figure}[htb]
  \centering
\setlength{\tabcolsep}{0in}
  $\begin{tabular}{cc}
\multicolumn{1}{m{.5\columnwidth}}{\centering\includegraphics[width=.3\columnwidth]{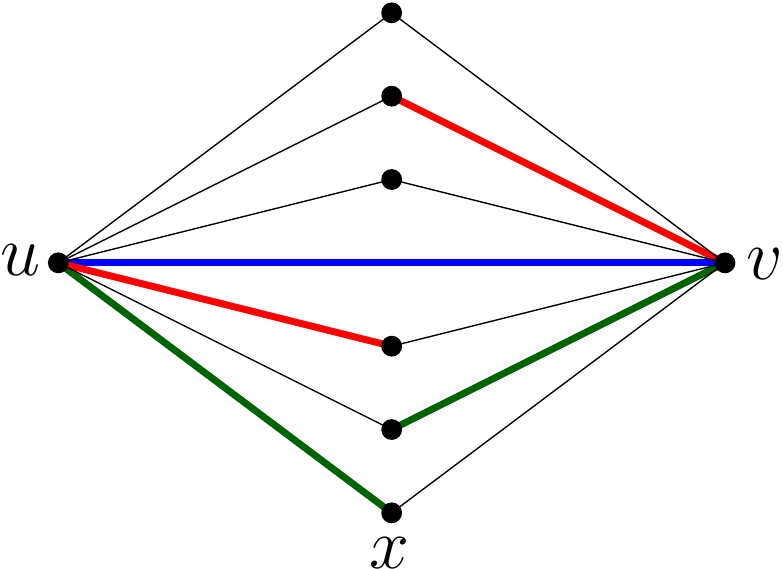}}&
\multicolumn{1}{m{.5\columnwidth}}{\centering\includegraphics[width=.15\columnwidth]{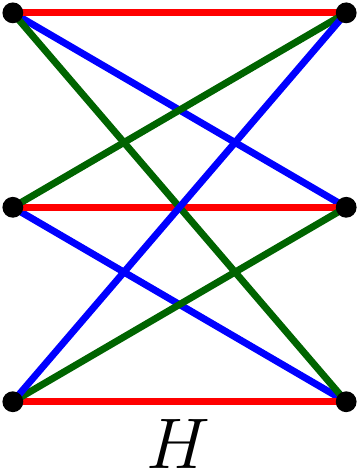}}\\
(a) & (b)
\end{tabular}$
  \caption{(a) By removing any matching (red, green, or blue) from $K_n$, at most two paths between $u$ and $v$ disappear. (b) The edges of $H$ are partitioned to $\frac{n}{2}$ perfect matchings, thus, $K_n-H$ does not have any perfect matching.}
\label{path-matching}
\end{figure}

\begin{lemma}
\label{regularity-connectivity}
Let $K_n$ be a complete graph with $n$ vertices, where $n$ is even, and let $\mathcal{M}$ be a set of $k$ edge-disjoint perfect matchings in $K_n$. Then, $K_n-\mathcal{M}$ is an $(n-1-k)$-regular graph which is $(n-1-2k)$-connected.
\end{lemma}
\begin{proof}
The regularity is trivial, because $K_n$ is $(n-1)$-regular and every vertex has degree $k$ in $\mathcal{M}$, thus, $K_n-\mathcal{M}$ is an $(n-1-k)$-regular graph. Now we prove the connectivity. Consider two vertices $u$ and $v$ in $V(K_n)$. There are $n-1$ many edge-disjoint paths between $u$ and $v$ in $K_n$: $n-2$ many paths of length two of the form $(u,x,v)$, where $x\in V(K_n)-\{u,v\}$ and a path $(u,v)$ of length one; see Figure~\ref{path-matching}(a). By removing any matching in $\mathcal{M}$ from $K_n$, at most two paths disappear, because $u$ and $v$ have degree one in each matching. Thus in $G-\mathcal{M}$, there are $(n-1-2k)$ many edge-disjoint paths between $u$ and $v$, which implies that $K_n-\mathcal{M}$ is $(n-1-2k)$-connected.
\end{proof}

\begin{theorem}
\label{mp-thr}
 $mp(K_n)\ge\frac{n}{2}$.
\end{theorem}
\begin{proof}
In order to prove the theorem, we show that by removing any set $\mathcal{M}$ of $k$ edge-disjoint perfect matchings from $K_n$, where $0\le k<\frac{n}{2}$, the resulting graph still has a perfect matching. By Lemma~\ref{regularity-connectivity}, the graph $K_n-\mathcal{M}$ is an $(n-1-k)$-regular graph which is $(n-1-2k)$-connected. Since $k<\frac{n}{2}$, $K_n-\mathcal{M}$ is a connected graph and the degree of each vertex is at least $\frac{n}{2}$. Thus, by a result of 
Dirac~\cite{Dirac1952}, $K_n-\mathcal{M}$ has a Hamiltonian cycle and consequently a perfect matching. Therefore, by removing $k$ arbitrary perfect matchings from $K_n$, where $k<\frac{n}{2}$, the resulting graph still has a perfect matching, which proves the claim.
\end{proof}

\begin{lemma}
\label{bipartite-matchings-lemma}
 If $n\equiv 2 \mod 4$, then $mp(K_n)\le\frac{n}{2}$.
\end{lemma}
\begin{proof}
Let $H=K_{\frac{n}{2},\frac{n}{2}}$ be a complete bipartite subgraph of $K_n$. Note that $\frac{n}{2}$ is an odd number and $H$ is an $\frac{n}{2}$-regular graph.
According to Hall's marriage theorem \cite{Hall1935}, for $k\ge 1$, every $k$-regular bipartite graph contains a perfect matching \cite{Harary1991}. Since by the iterative removal of perfect matchings  from $H$ the resulting graph is still regular, the edges of $H$ can be partitioned into $\frac{n}{2}$ perfect matchings; see Figure~\ref{path-matching}(a). It is obvious that $K_n-H$ consists of two connected components of odd size. Thus, by removing the $\frac{n}{2}$ matchings in $H$, the resulting graph, $K_n-H$, does not have any perfect matching. This proves the claim.
\end{proof}

By Theorem~\ref{mp-thr} and Lemma~\ref{bipartite-matchings-lemma} we have the following corollary.

\begin{corollary}
If $n\equiv 2 \mod 4$, then $mp(K_n)=\frac{n}{2}$. 
\end{corollary}

In the rest of this section we consider plane matching removal from geometric graphs.

Let $P$ be a set of $n$ points in general position in the plane, with $n$ even. Given a geometric graph $G$ on $P$, we say that $G$ is {\em plane matching persistent} if by removing any plane perfect matching $M$ from $G$, the resulting graph, $G-M$, has a plane perfect matching. We define the {\em plane matching persistency} of $G$, denoted by $pmp(G)$ as the size of the smallest set $\mathcal{M}$ of edge-disjoint plane perfect matchings that can be removed from $G$ such that $G-\mathcal{M}$ does not have any plane perfect matching. In particular, $G$ is plane matching persistent iff $pmp(G)\ge 2$.

Aichholzer et al.~\cite{Aichholzer2010} and Perles (see~\cite{Keller2012}) showed that by removing any set of at most $\frac{n}{2}-1$ edges from $\Kn{(P)}$, the resulting graph has a plane perfect matching. This bound is tight~\cite{Aichholzer2010}; that is, there exists a point set $P$ such that by removing a set $H$ of $\frac{n}{2}$ edges from $\Kn{(P)}$ the resulting graph does not have any plane perfect matching. In the examples provided by~\cite{Aichholzer2010}, the $\frac{n}{2}$ edges in $H$ form a connected component which has $\frac{n}{2}+1$ vertices. 

Thus, one may think if the removed edges are disjoint, it may be possible to remove more than $\frac{n}{2}-1$ edges while the resulting graph has a plane perfect matching.
In the following lemma we show that by removing any plane perfect matching, i.e., a set of $\frac{n}{2}$ disjoint edges, from $\Kn{(P)}$, the resulting graph still has a perfect matching.

\begin{lemma}
\label{pmp2-lemma}
Let $P$ be a set of $n$ points in general position in the plane with $n$ even, then $pmp(\Kn{(P)})\ge 2$.
\end{lemma}
\begin{proof}
 Let $M$ be any plane perfect matching in $\Kn{(P)}$. Assign $\frac{n}{2}$ distinct colors to the points in $P$ such that both endpoints of every edge in $M$ have the same color. By Theorem~\ref{Aichholzer}, $P$ has a plane colored matching, say $M'$. Since both endpoints of every edge in $M$ have the same color while the endpoints of every edge in $M'$ have distinct colors, $M$ and $M'$ are edge-disjoint. Therefore, by removing any plane perfect matching from $\Kn{(P)}$, the resulting graph still has a plane perfect matching, which implies that $pmp(\Kn{(P)})\ge 2$.
\end{proof}

\begin{theorem}
For a set $P$ of $n\ge 4$ points in convex position in the plane with $n$ even, $pmp(\Kn{(P)})=2$.
\end{theorem}
\begin{proof}
By Lemma~\ref{pmp2-lemma}, $pmp(\Kn{(P)})\ge2$. In order to prove the theorem, we need to show that $pmp(\Kn{(P)})\le2$. Let $M_1$ and $M_2$ be two edge-disjoint plane matchings obtained from \CH{P}. By Lemma~\ref{two-convex-edges}, any plane perfect matching in $\Kn{(P)}$ contains at least two edges of \CH{P}, while $\Kn{(P)}-\{M_1\cup M_2\}$ does not have convex hull edges, and hence does not have any plane perfect matching. Therefore, $pmp(\Kn{(P)})\le2$. 
\end{proof}

\begin{lemma}
There exists a point set $P$ in general position such that $pmp(K(P))\ge 3$.
\end{lemma}
\begin{proof}
We prove this lemma by providing an example. Figure~\ref{pmp3-fig}(a) shows a set $P=\{a_1,\dots,a_n,\allowbreak b_1,\dots,b_n,\allowbreak c_1,\dots,c_n\}$ of $3n$ points in general position, where $n$ is an even number. In order to prove that $pmp(K(P))\ge 3$, we show that by removing any two edge-disjoint plane matchings from $K(P)$, the resulting graph still has a plane perfect matching. Let $M_1$ and $M_2$ be any two plane perfect matchings in $K(P)$. Let $G$ be the subgraph of $K(P)$ induced by the edges in $M_1\cup M_2$. Note that $G$ is a 2-regular graph and by Observation~\ref{even-cycle-obs} does not contain any odd cycle. For each $1\le i\le n$, let $t_i$ be the triangle which is defined by the three points $a_i$, $b_i$, and $c_i$. Let $\mathcal{T}$ be the set of these $n$ (nested) triangles. Since $G$ does not have any odd cycle, for each $t_i\in\mathcal{T}$, at least one edge of $t_i$ is not in $G$. Let $M_3$ be the matching containing an edge $e_i$ from each $t_i\in \mathcal{T}$ such that $e_i\notin G$. See Figure~\ref{pmp3-fig}(b). Now we describe how to complete $M_3$, i.e., complete it to a perfect matching. Partition the triangles in $\mathcal{T}$ into $\frac{n}{2}$ pairs of consecutive triangles. For each pair $(t_i,t_{i+1})$ of consecutive triangles we complete $M_3$ locally\textemdash on $a_i,b_i,c_i,a_{i+1},b_{i+1},c_{i+1}$\textemdash in the following way. Let $t_i=(a_i,b_i,c_i)$ and $t_{i+1}=(a_{i+1},b_{i+1},c_{i+1})$. See Figure~\ref{pmp3-fig}(c). W.l.o.g. assume that $M_3$ contains $a_ib_i$ and $a_{i+1}c_{i+1}$, that is $a_ib_i\notin G$ and $a_{i+1}c_{i+1}\notin G$. If $c_ib_{i+1}\notin G$, then we complete $M_3$ by adding $c_ib_{i+1}$. If $c_ib_{i+1}\in G$, then $a_{i+1}b_{i+1}\notin G$ or $c_{i+1}b_{i+1}\notin G$ because $b_{i+1}$ has degree two in $G$. W.l.o.g. assume that $a_{i+1}b_{i+1}\notin G$. Then we modify $M_3$ by removing $a_{i+1}c_{i+1}$ and adding $a_{i+1}b_{i+1}$. Now, if $c_ic_{i+1}\notin G$, then we complete $M_3$ by adding $c_ic_{i+1}$. If $c_ic_{i+1}\in G$, then by Observation~\ref{even-cycle-obs}, $b_{i+1}c_{i+1}\notin G$. We modify $M_3$ by removing $a_{i+1}b_{i+1}$ and adding $b_{i+1}c_{i+1}$. At this point, since $c_ib_{i+1}$ and $c_ic_{i+1}$ are in $G$, $c_ia_{i+1}\notin G$ and we complete $M_3$ by adding $c_ia_{i+1}$.
\end{proof}
\begin{figure}[htb]
  \centering
\setlength{\tabcolsep}{0in}
  $\begin{tabular}{ccc}
\multicolumn{1}{m{.33\columnwidth}}{\centering\includegraphics[width=.3\columnwidth]{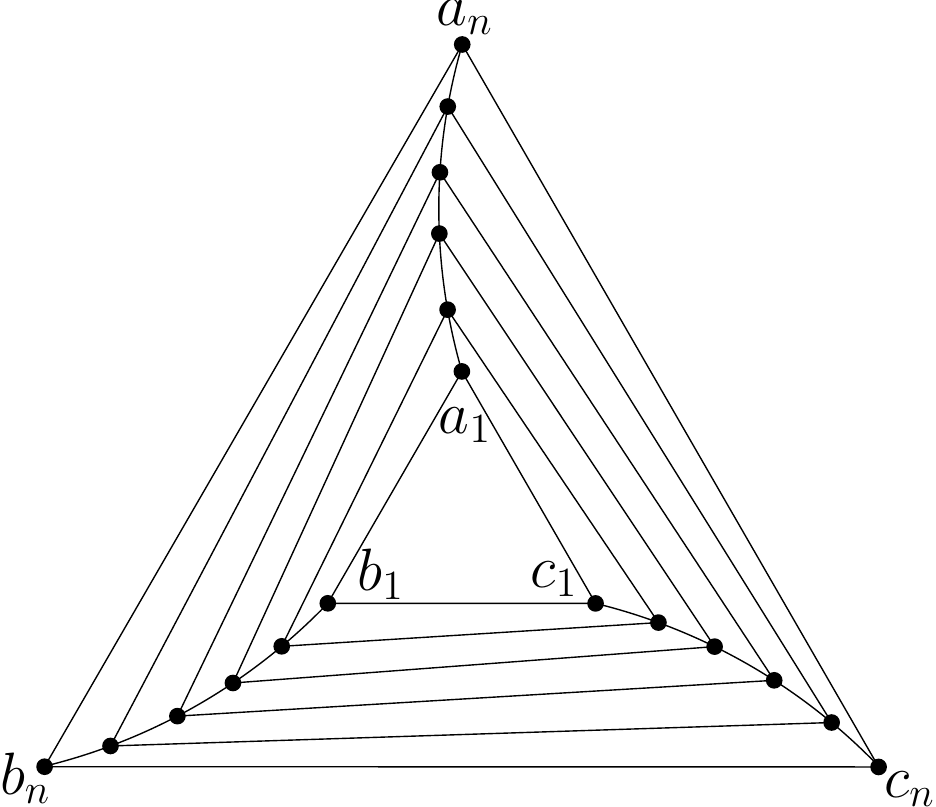}}&
\multicolumn{1}{m{.33\columnwidth}}{\centering\includegraphics[width=.3\columnwidth]{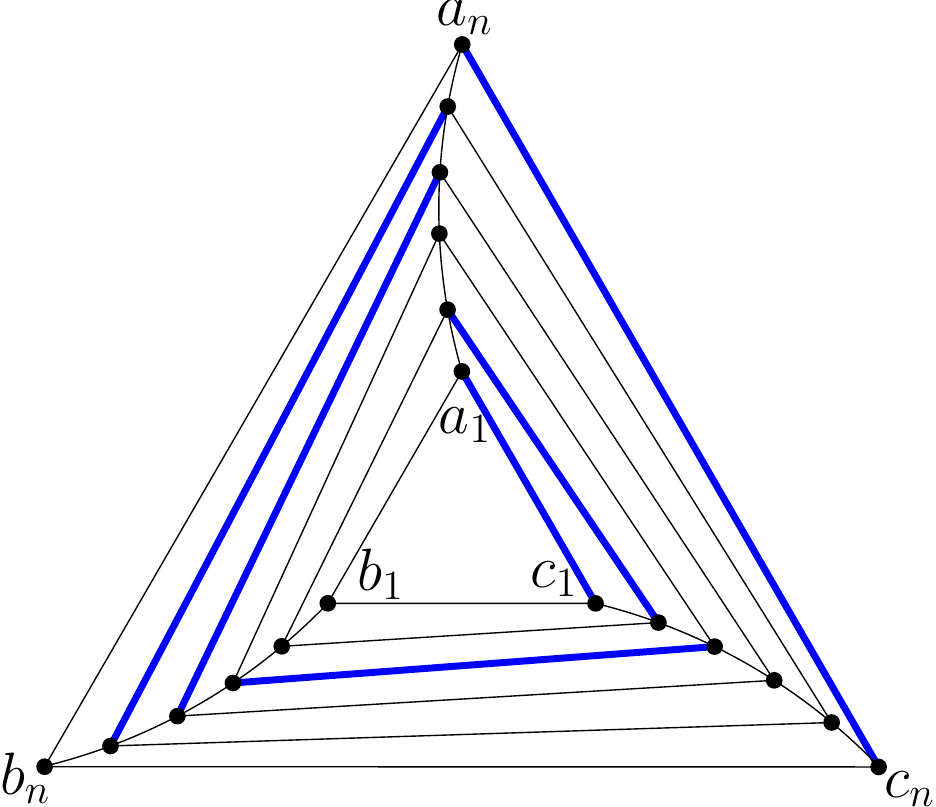}}
&\multicolumn{1}{m{.33\columnwidth}}{\centering\includegraphics[width=.3\columnwidth]{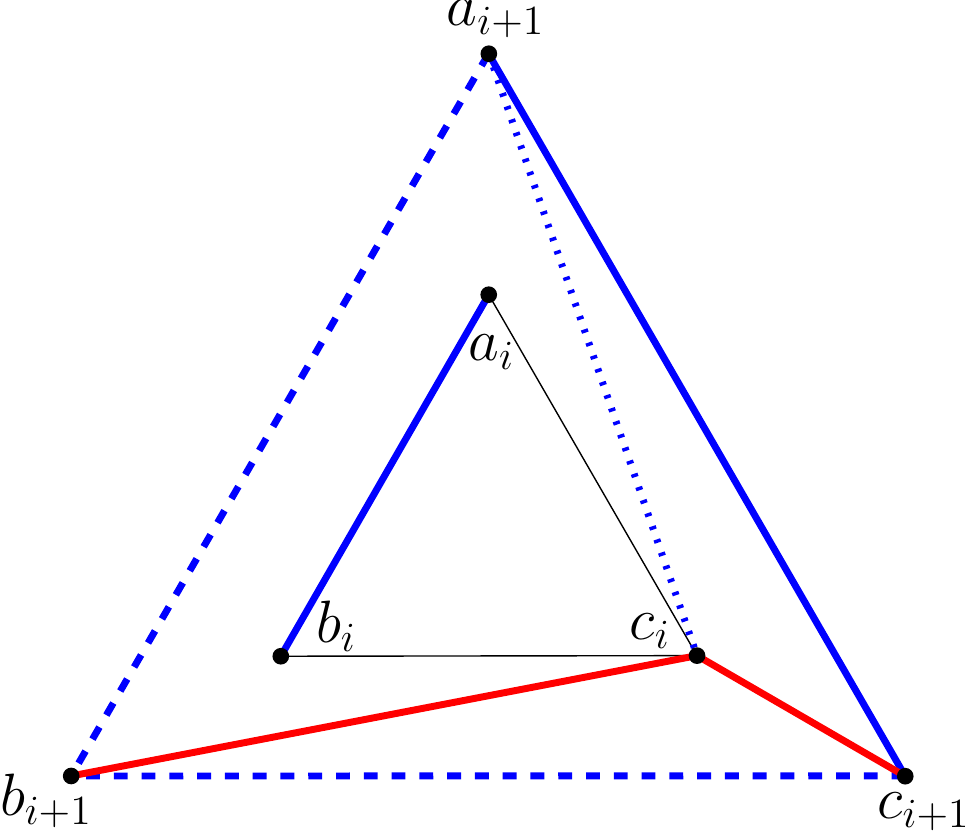}} \\
(a) & (b)&(c)
\end{tabular}$
\caption{(a) Set $P$ of $3n$ points in general position. (b) $M_3$ contains one edge from each triangle. (c) Locally converting $M_3$ to a perfect matching, for $t_i$ and $t_{i+1}$.}
\label{pmp3-fig}
\end{figure}

\section{Conclusion}
\label{conclusion}
In this paper, we considered the problem of packing edge-disjoint plane perfect matchings in a complete geometric graph $\Kn{(P)}$ on a set $P$ of $n$ points in general position in the plane. We proved that
\begin{itemize}
  \item at least $\lceil\log_2{n}\rceil-2$ plane matchings can be packed into $\Kn{(P)}$,
  \item at least two and at most five non-crossing plane matchings can be packed into $\Kn{(P)}$,
  \item the plane matching persistency of $\Kn{(P)}$ is at least two.
\end{itemize}
In addition, for some special configurations of $P$ we showed that
\begin{itemize}
  \item if $P$ is in convex position, then $\frac{n}{2}$ plane matchings can be packed into $\Kn{(P)}$,
  \item if $P$ is in a regular wheel configuration, then $\frac{n}{2}-1$ plane matchings can be packed into $\Kn{(P)}$,
  \item there exists a set $P$ such that no more than $\lceil\frac{n}{3}\rceil$ plane matchings can be packed into $\Kn{(P)}$,
  \item if $P$ is in convex position, the plane matching persistency of $\Kn{(P)}$ is two,
  \item there exists a set $P$ such that the plane matching persistency of $\Kn{(P)}$ is at least three.
\end{itemize}

We leave a number of open problems:
\begin{itemize}
  \item We believe that the number of plane matchings that can be packed into $\Kn{(P)}$ is linear in $n$. Thus, improving the lower bound of $\lceil\log_2{n}\rceil-2$ is the main open problem.
  \item Is there an upper bound better than $n-1$, on the number of plane matchings that can be packed into $\Kn{(P)}$, where $n>6$?
  \item Providing point sets with large plane matching persistency.
\end{itemize}
 
\bibliographystyle{abbrv}
\bibliography{Matching-Packing.bib}
\end{document}